\documentclass{article}[12pt]
\usepackage{float}
\usepackage[english]{babel}
\usepackage[margin=1in]{geometry}
\usepackage{amssymb}
\usepackage{amsthm}
\usepackage{mathtools}
\usepackage{bbm, amsmath, amsfonts}
  \usepackage{paralist}
  \usepackage{graphics} 
  \usepackage{epsfig} 
\usepackage{graphicx}  
\usepackage{epstopdf}
\usepackage[font={small,it}]{caption}
\usepackage{subcaption}
\graphicspath{{figures/}} 

 \usepackage[colorlinks=true]{hyperref}
\hypersetup{urlcolor=blue, citecolor=red}

\newcommand{\rmd}{\mathrm{d}}

\renewcommand{\epsilon}{\varepsilon}

\newcommand{\one}{{\mathbbm{1}}}
\newcommand{\nwc}{\newcommand}
\nwc{\calA}{{\mathcal A}}
\nwc{\ip}[1]{{\langle #1 \rangle}}
\nwc{\inv}{^{-1}}
\nwc{\D}{\partial}

\nwc{\red}[1]{\textcolor{red}{#1}}
\nwc{\blue}[1]{\textcolor{blue}{#1}}
 
\newtheorem{theorem}{Theorem}[section]
\newtheorem{corollary}[theorem]{Corollary}
\newtheorem{lemma}[theorem]{Lemma}
\newtheorem{proposition}[theorem]{Proposition}

\theoremstyle{definition}

\newtheorem{remark}[theorem]{Remark}

\numberwithin{equation}{section}

\date{\today}

\begin{document}
\author{Pierre Degond\thanks{Department of Mathematics, Imperial College London,
London SW7 2AZ, United Kingdom. pdegond@imperial.ac.uk}, Maximilian Engel\thanks{Department of Mathematics, Technical University of Munich, Munich D-85748, Germany. maximilian.engel@tum.de}, Jian-Guo Liu\thanks{Department of Physics and Department of Mathematics, Duke University, Durham, NC 27708, USA. Jian-Guo.Liu@duke.edu}, Robert L. Pego\thanks{Department of Mathematics, Carnegie Mellon University, Pittsburgh, Pennsylvania, PA 12513, USA. rpego@cmu.edu}}

\title{A Markov jump process modelling animal group size statistics}

\maketitle

\begin{abstract}
We translate a coagulation-framentation model, describing the dynamics of
animal group size distributions, into a model for the population distribution
and associate the nonlinear evolution equation with a Markov jump process
of a type introduced in classic work of H.~McKean.
In particular this formalizes a model suggested by H.-S. Niwa
[J.~Theo.~Biol.~224 (2003)] with simple coagulation and fragmentation rates.
Based on the jump process, we develop a numerical scheme that allows us to
approximate the equilibrium for the Niwa model, validated by comparison to
analytical results by Degond et al.  [J.~Nonlinear Sci.~27 (2017)], and study
the population and size distributions for more complicated rates. Furthermore,
the simulations are used to describe statistical properties of the underlying
jump process. We additionally discuss the relation of the jump process to
models expressed in stochastic differential equations and demonstrate that such
a connection is justified in the case of nearest-neighbour interactions, as
opposed to global interactions as in the Niwa model.
\end{abstract}

\bigskip

{\bf Keywords:} Population dynamics, numerics, jump process, fish schools, {self-consistent} Markov process

{\bf Mathematics Subject Classification (2010):} 60J75, 65C35, 70F45, 92D50, 45J05, 65C30

\section{Introduction} \label{intro}
The aggregation of animals into groups of different sizes involves a range of stimulating mathematical problems. On the one hand, the changes in size of the group a certain individual belongs to can be described, on the microscopic level, as a (stochastic) process. On the other hand, this process can be associated, on the macroscopic level, to the distribution of group sizes and the probability of individuals to belong to a group of a certain size (referred to below as the population distribution), and the evolution of such distributions in time. In particular, the existence and uniqueness of and convergence to an equilibrium distribution is a central object of interest. 

Various models of describing the coagulation and fragmentation of groups of animals have been suggested and analysed in the past (cf. e.g. \cite{BD, BDF, G, GL, O}). The model this work rests upon was introduced by Hiro-Sato Niwa in 2003 \cite{N} related to studies in \cite{N1, N2, N3} and has turned out to hold for data from pelagic fish and mammalian herbivores in the wild. 
Niwa simulated a very simple merge and split process for a fixed population but he did not analyze the actual process he simulated. Instead, he used kinetic Monte-Carlo simulations to fit the noise in a stochastic differential equation (SDE) model for the size of the group that an individual animal belongs to. Due to its fairly simple form, he was able to find a closed formula for the stationary density of this SDE in a form similar to an exponential law (modified by a double-exponential factor), interpreting it as the equilibrium population distribution. Since the population distribution is related to the group-size distribution by a simple algebraic relation, he was able to find the equilibrium group-size distribution in a form close to an inverse power law with an exponential cutoff, namely 
\begin{equation}
    \Phi_\star(x) \propto x\inv \exp\left( -x+ \frac12 x e^{-x}\right).
\label{Phi-Niwas}
\end{equation}
In \cite{N}, Niwa showed that this expression 
provided a good fit to a large amount of empirical data with no fitting parameters.

In \cite{DLP}, we have pursued Niwa's model based on a coagulation-fragmentation formulation for the distribution of group sizes and given a rigorous description of the equilibria for continuous and discrete cluster sizes. The lack of a detailed balanced condition presented a true mathematical challenge. However, by introducing the so-called Bernstein transformation, we have shown that there exists a unique equilibrium, under a suitable normalization condition, for both the discrete and the continuous cluster size cases. Furthermore, we provided numerical investigations of the model in \cite{DE}.

In the present paper we develop Niwa's original idea for modelling the population distribution through the group-size history of a fixed individual. We derive and study the naturally associated jump process rather than using the SDE framework, however.
The jump process is described through a self-consistent Markovian approach
as introduced in classic work of H.~McKean~\cite{McKean66},
in which the rates of jumps for a (tagged) individual's group size depend upon
the (macroscopic) group-size distribution for the whole population.  
To be a self-consistent description of the dynamics, this macroscopic
group-size distribution should coincide with the probability distribution
that evolves under the jump process for the tagged individual.

This feedback makes it difficult to handle such a jump process analytically
in the cases we wish to consider. (Some of the earliest analytic work on 
a general class of jump processes that follow McKean's framework was 
carried out by Ueno \cite{Ueno} and Tanaka \cite{Tanaka1970a, Tanaka1970b}. 
A vast literature now exists on related McKean-Vlasov diffusion processes.)
Therefore we develop an algorithm that approximates the jump process
and which we can use to study the dynamics of the process 
and the equilibrium group-size distribution for variations of the Niwa model.
The idea is to estimate the macroscopic group-size distribution 
for a large population of total size $N$ 
by the empirical distribution of a sample of $\tilde N$ tagged individuals
with $\tilde N\ll N$.  This effectively results in a Markov jump process
for the group sizes of a fixed number $\tilde N$ of tagged individuals, with transition
rates driven by the population distribution of these individuals.
A closely related time-continuous Markovian interacting particle system was 
studied by Eibeck and Wagner in \cite{EW2000} and \cite{EW2001}
where they proved convergence of the empirical distribution to 
a solution of coagulation models which was extended to coagulation-fragmentation models, among others \cite{EW2003}. In this paper we will pay particular attention, however, not only
to the estimated macroscopic group-size distribution, but to 
statistical and dynamical properties of the jump process 
for a tagged individual. 

We refer to \cite{Aldous99} for an early review of open problems relating
Smoluchowski's coagulation equations to stochastic models, including
the standard Marcus-Lushnikov process \cite{Marcus68, Lushnikov78} 
for the size distribution of all groups in a fixed finite population.  
Our approach and the one in \cite{EW2001} deal with somewhat simpler 
jump-process models for the population distribution for a 
fixed number of groups (which can be regarded as those containing 
tagged individuals).

In \cite{DE}, we approached the coagulation-fragmentation form of the Niwa model with three different numerical methods.
One of them is a recursive algorithm derived from the discrete model in \cite{MJS} which allows to compute the equilibrium precisely but only for fixed coagulation and fragmentation rates being constant over time and group sizes. It further does not give insight into the time-evolution.  
The second method is a Newton-like method which also only allows to approximate the equilibrium, but with the advantage of being adaptable to size-dependent merge and split rates. 
The third method is a time-dependent Euler method which is flexible towards arbitrary modifications of the model but numerically not efficient. 

We will show in the current paper that our numerical scheme based on a Markov chain can be used to approximate the equilibrium for all different kinds of coagulation and fragmentation rates, allows for insights into the time evolution of the population distribution and also enables us to study  properties of the process such as the statistics of occupation times in cluster sizes and the decay of correlations of trajectories. The method is accurate, efficient (in particular in equilibrium where it is sufficient to simulate only one trajectory), versatile, and gives insights into the dynamics on the individual level.

Furthermore, we will also discuss some shortcomings of Niwa's SDE model for the 
temporal behavior of the size of the group containing a given individual.
As merging and splitting occurs, 
involving interactions of groups of all different sizes,
an individual's group size can be expected to experience frequent large jumps.
Yet the SDE model is an Ornstein-Uhlenbeck-type (OU-type) process that
has continuous paths in time. 
One way this could be a reasonable approximation is if
the jump process has mostly small jumps, for then there
is a natural SDE approximation found
through a second-order Taylor expansion.
We find, though, that for the jump process in question, 
the equilibrium of the second-order SDE approximation is not consistent with
the rigorously derived equilibrium in \cite{DLP}. 

Another modeling issue is that the OU-type process will naturally
produce unphysical (negative) group sizes. Some kind of 
reflection or symmetrization is needed to maintain positive sizes, 
but the natural choice leaves a free parameter in the model
and is not motivated well by merging/splitting mechanisms.
Niwa's model also involves an exponentially growing variance which 
makes accurate numerical simulation difficult.

Lastly, we will point out that there is a kind of group-size dynamics,
distinctly different from what Niwa simulated, 
that admits an SDE model whose equilibria are quite similar to 
the equilibria found in \cite{N} and \cite{DLP}, but even simpler.
The equilibria take exactly the form of a simple power law with exponential 
cutoff---a logarthimic or gamma distribution.
This alternative SDE model for such distributions
goes back to \cite{Maybook} and \cite{Polansky}, 
and corresponds to a process with continuous sample paths
guaranteed to stay in $(0, \infty)$ without hitting 0. 
We show formally the convergence of a nearest-neighbour model 
for jumps in group sizes to the equilibrium of this SDE.

The remainder of this paper is structured as follows. In Section~\ref{sec:jumpprocess}, we derive the evolution equation of the population distribution for a coagulation-fragmentation model of size distributions, in general, and the Niwa model, in particular. 
Furthermore, this nonlinear evolution equation is shown to coincide with the master equation for a 
Markov jump process whose jump rates are determined self-consistently by the population distribution itself. 
Section~\ref{algorithm} introduces the numerical scheme which is then used in Section~\ref{nuumval} to simulate the process for different choices of coagulation and fragmentation rates. We validate the method by observing fast and accurate approach of the (known) equilibrium for the Niwa model and further use the algorithm to generate the size distributions for random and polynomial rates. In Section~\ref{sec:statistics}, we use the numerical method to estimate the decay of correlations of the Niwa jump process and, additionally, describe the statistics of ocuupation times for different kinds of rates. Finally, Section~\ref{sec:SDEs} is dedicated to the role of stochastic differential equations in the context of the aggregation models.

\setcounter{equation}{0}
\section{Description of the (self-consistent) Markov jump process for the population distribution} \label{sec:jumpprocess}
\subsection{Evolution equation for the population distribution}
\label{sec:description}
The continuous version of a coagulation-fragmentation equation, called also Smoluchowski equation, describes the evolution of the number density $f(x,t)$ of continuous sizes $x \geq 0$ at time $t$. In weak form it reads, for all test functions $\varphi \in C((0, \infty))$: 
\begin{equation} \label{weakgeneral}
\begin{split}
\frac{\rmd}{\rmd t}\int_{\mathbb{R}_{+}} \varphi(s)f(s,t) \rmd s
= \frac{1}{2} \int_{(\mathbb{R}_{+})^2} (\varphi(s+\hat{s}) - \varphi(s) - \varphi(\hat{s}))a(s,\hat{s})f(s,t)f(\hat{s},t) \rmd s \, \rmd \hat{s} \\
- \frac{1}{2} \int_{(\mathbb{R}_{+})^2}(\varphi(s+\hat{s}) - \varphi(s) - \varphi(\hat{s}))b(s,\hat{s}) f(s+\hat{s},t) \rmd s \, \rmd \hat{s}.
\end{split}
\end{equation}
The coagulation rate $a(s,\hat{s})$ and fragmentation rate $b(s,\hat{s})$ are both nonnegative and symmetric. The coagulation and fragmentation reactions can be written schematically
\begin{align*}
(s) + (\hat{s}) \ &\xrightarrow{a(s,\hat{s})} \ (s+\hat{s}) \quad \text{(binary coagulation)},\\
(s) + (\hat{s}) \ &\xleftarrow{b(s,\hat{s})} \ (s+\hat{s}) \quad \text{(binary fragmentation)}.
\end{align*}
By a change of variables, (\ref{weakgeneral}) can be transformed into 
\begin{equation} \label{weakgeneralalt}
\begin{split}
\frac{\rmd}{\rmd t}\int_{\mathbb{R}_{+}} \varphi(s)f(s,t) \rmd s
= \frac{1}{2} \int_{(\mathbb{R}_{+})^2} (\varphi(s+\hat{s}) - \varphi(s) - \varphi(s))a(s,\hat{s})f(s,t)f(\hat{s},t) \rmd s \, \rmd \hat{s} \\
- \frac{1}{2} \int_{(\mathbb{R}_{+})}\left( \int_{0}^s(\varphi(s) - \varphi(\hat{s}) - \varphi(s-\hat{s}))b(\hat{s},s-\hat{s}) \, \rmd \hat{s} \right) f(s,t) \, \rmd s.
\end{split}
\end{equation}
Note that by taking $\varphi(s) = s$, one obtains the conservation of mass
\begin{equation} \label{mass}
\frac{\rmd}{\rmd t} \int_{\mathbb{R}_{+}} s f(s,t) \, \rmd s = 0.
\end{equation}
Starting from the group size distrbution $f (s,t)$ satisfying equation~\eqref{weakgeneralalt},
we introduce the population distribution 
\begin{equation} \label{PopDistr}
\rho(s,t) = \frac{s \, f(s,t)}{\int_{{\mathbb R}_+} s \, f(s,t) \, \rmd s}\,.
\end{equation}
From \eqref{mass} we observe that
$$ \int_{{\mathbb R}_+} s \, f(s,t) \, \rmd s =: N $$ 
is conserved and corresponds to the total number of individuals $N$. 

In Niwa's model \cite{N}, the coagulation and fragmentation rates are constant. The setting of the model assumes different zones of space on which $N$ individuals move, where $N$ is conserved through time. At each time step every group, whose size is a natural number $i \in \mathbb{N}$, moves towards a randomly selected site with equal probability. When $i$- and $j$-sized groups meet at the same site, they aggregate to a group of size $i+j$. So the coagulation rate is independent from the group sizes and can be written as $a_{i,j} = 2 \tilde q$ for any $i,j > 0$ where $\tilde q > 0$ is the fixed coagulation parameter. 
The fragmentation rate $b_{i,j}$  expresses the fact that at each time step each group with size $k \geq 2$ splits with probability $\tilde p$ independent of $k$, and that if it does split, it breaks into one of the pairs with sizes $ (1,k-1), (2,k-2), \dots, (k-1,1)$ with equal probability. As the actually distinct pairs are counted twice in such an enumeration, one gets for all $1 \leq i,j < k$ with $ i + j =k $: 
$b_{i,j} = \frac{\tilde p}{(i+j-1)/2} = \frac{2 \tilde p}{i+j-1}$.

The formulation for continuous cluster sizes gives
\begin{equation} \label{DegondrateC}
a(s,\hat{s}) = q, \quad b(s,\hat{s}) = \frac{p}{s+\hat{s}}\,,
\end{equation}
where $q= 2 \tilde{q}$  and $p = 2 \tilde{p}$ for $\tilde{q}$ and $\tilde{p}$ being the constants in Niwa's model. In this case the group size distribution $f (s,t)$ satisfies the following equation \cite{DLP}  
\begin{eqnarray}
&&\hspace{-1cm}
\frac{\rmd}{\rmd t} \int_{{\mathbb R}_+} \varphi(s) \, f(s,t) \, \rmd s =
\frac{q}{2} \int_{({\mathbb R}_+)^2} \big( \varphi (s+\hat s) - \varphi(s) - \varphi(\hat s) \big) \, f(s,t) \, f(\hat s,t)  \, \rmd s \, \rmd \hat s  \nonumber \\
&&\hspace{2cm}
- \frac{p}{2} \int_{({\mathbb R}_+)^2} \big( \varphi (s) - \varphi(\hat s) - \varphi(s - \hat s) \big) \,  \frac{f(s,t)}{s} \, \chi_{[0,s)} (\hat s) \, \rmd s \, \rmd \hat s \,.
\label{eq:CF1_Niwa}
\end{eqnarray}
In \cite[Theorem 5.1]{DLP}, the existence of a unique scaling profile $f_{*}$ for the equilibrium of~\eqref{eq:CF1_Niwa} is proven, where $f_*(x) = \gamma_*(x) e^{-\frac{4}{27} x}$ for all $x \in (0, \infty)$ and $\gamma_*$ is a \emph{completely monotone} function (infinitely differentiable with derivatives that alternate in sign) with asymptotic power laws.

We derive the general strong form for the population density corresponding with model~\eqref{weakgeneral}, presented in the following Proposition, and thereby in particular for the Niwa model, as stated in the subsequent Corollary. 
\begin{proposition} \label{prop:Evequ}
In strong form the evolution equation of the population density $\rho(\cdot,t)$ at any time $t \in \mathbb{R}_+$, corresponding with~\eqref{weakgeneral}, is given by
\begin{align}
\partial_t \rho(s,t) &=
 N\, \int_0^\infty  \Big( a(\hat{s}, s - \hat{s})\frac{\rho(\hat s,t)  \, \rho(s - \hat s ,t)}{s - \hat s} \chi_{[0,s)}(\hat s)    - a(s,\hat{s})\frac{\rho(s,t) \, \rho(\hat s,t)}{\hat s} \Big)  \, \rmd \hat s \nonumber \\
&+   s \int_s^\infty \frac{b(s, \hat s - s)}{\hat s} \rho(\hat s, t) \, \rmd \hat s -  \rho(s,t) \int_{0}^{\infty} b(\hat s, s - \hat s)\chi_{[0,s/2]} (\hat s)  \,  \rmd \hat s . 
\label{eq:jumpproc_stronggen}
\end{align}
\end{proposition}
\begin{proof}
By testing \eqref{weakgeneralalt} against $s \varphi(s)/N$ we obtain:
\begin{align*}
\frac{\rmd}{\rmd t} \int_{{\mathbb R}_+} \varphi(s) \, \rho(s,t) \, \rmd s &=
\int_{({\mathbb R}_+)^2}\frac{a(s,\hat{s})}{2N}  \big( (s+\hat s)\varphi (s+\hat s) - s\varphi(s) - \hat s \varphi(\hat s) \big) \, f(s,t) \, f(\hat s,t)  \, \rmd s \, \rmd \hat s  \nonumber \\
&-  \int_{({\mathbb R}_+)^2} \frac{b(\hat{s}, s-\hat{s})}{2N} \big( s \varphi (s) - \hat s \varphi(\hat s) - (s - \hat s) \varphi(s - \hat s) \big) \,  f(s,t) \, \chi_{[0,s)} (\hat s) \, \rmd s \, \rmd \hat s \\
&=  \int_{({\mathbb R}_+)^2} \frac{N a(s,\hat{s})}{2} \big( (s+\hat s)\varphi (s+\hat s) - s\varphi(s) - \hat s \varphi(\hat s) \big) \, \frac{\rho(s,t) \, \rho(\hat s,t)}{s \hat s}  \, \rmd s \, \rmd \hat s  \nonumber \\
&-  \int_{({\mathbb R}_+)^2} \frac{b(\hat{s}, s-\hat{s})}{2} \big( s \varphi (s) - \hat s \varphi(\hat s) - (s - \hat s) \varphi(s - \hat s) \big) \,  \frac{\rho(s,t)}{s} \, \chi_{[0,s)} (\hat s) \, \rmd s \, \rmd \hat s.
\end{align*}
We remark that 
$$ (s+\hat s)\varphi (s+\hat s) - s\varphi(s) - \hat s \varphi(\hat s)  = s ( \varphi (s+\hat s) - \varphi(s)) + \hat s ( \varphi (s+\hat s) - \varphi(\hat s))\,, $$
and by symmetry under exchanges of $s$ and $\hat s$, the first integral becomes 
\begin{equation} \label{massintegral}
\int_{({\mathbb R}_+)^2} ( \varphi (s+\hat s) - \varphi(s))  \, \rho(s,t) \, \Big( a(s, \hat{s}) N \frac{\rho(\hat s,t)}{\hat s} \Big)   \, \rmd s \, \rmd \hat s.
\end{equation} 
Noting that the change of variables $\hat s \to s - \hat s$ leaves the second integral invariant, we can restrict the interval of integration in $\hat s$ to the interval $[0,s/2]$ upon multiplying the result by $2$. So, the second integral equals:
$$ \int_{({\mathbb R}_+)^2} \big(\frac{\hat s}{s} \varphi(\hat s) + \frac{s - \hat s}{s} \varphi(s - \hat s) - \varphi (s)\big) \,   b(\hat{s}, s - \hat{s}) \, \chi_{[0,s/2]} (\hat s) \rho(s,t) \, \rmd s \, \rmd \hat s. 
$$
The resulting equation is thus
\begin{align}
\frac{\rmd}{\rmd t} \int_{\mathbb{R}_+} \varphi(s) \rho(s,t) \rmd s &=  \int_{({\mathbb R}_+)^2} ( \varphi (s+\hat s) - \varphi(s))  \, \rho(s,t) \, \Big( a(s, \hat{s}) N \frac{\rho(\hat s,t)}{\hat s} \Big)   \, \rmd s \, \rmd \hat s \nonumber \\
&+  \int_{({\mathbb R}_+)^2} \big(\frac{\hat s}{s} \varphi(\hat s) + \frac{s - \hat s}{s} \varphi(s - \hat s) - \varphi (s)\big) \, \left( b(\hat{s}, s - \hat{s}) \, \chi_{[0,s/2]} (\hat s) \right) \rho(s,t) \, \rmd s \, \rmd \hat s.
\label{eq:jumpproc_weakgen}
\end{align}
To derive the strong form for the first integral in~\eqref{eq:jumpproc_weakgen} we only need to conduct a change of variables for the term involving $\varphi(s + \hat s)$, namely $ \hat s \to \hat s -s$:
$$  \int_{({\mathbb R}_+)^2}  \varphi (s+\hat s)a(s, \hat{s}) \rho(s,t) N \frac{\rho(\hat s,t)}{\hat s} \, \rmd s \, \rmd \hat s = \int_{0}^{\infty} \int_{0}^{\hat s} a(s, \hat{s}-s) N \frac{\rho(s,t) \rho(\hat s -s,t)}{\hat s -s}  \, \rmd s \,\varphi(\hat s)\ \rmd \hat s \,.
$$
For the second integral, again by symmetry between $\hat s$ and $s - \hat s$, it is enough to observe that
$$ \int_{({\mathbb R}_+)^2} \frac{\hat s}{s} \varphi(\hat s) b(\hat{s}, s - \hat{s}) \, \chi_{[0,s/2]} (\hat s) \rho(s,t) \, \rmd s \, \rmd \hat s = \frac{1}{2} \int_{0}^{\infty} \hat s \int_{\hat s}^{\infty} \frac{b(\hat{s}, s - \hat{s})}{s}   \,  \rho(s,t) \, \rmd s \, \varphi(\hat s) \rmd \hat s \,. $$
This finishes the proof.
\end{proof}
Note from the weak form~\eqref{eq:jumpproc_weakgen} in the proof that the merge process is done with rate $ a(s, \hat{s}) N \frac{\rho(\hat s,t)}{\hat s}$ while the split process is done with rate $ b(\hat{s}, s - \hat{s})$. Due to its dependence on $\rho$, the merge process is characterized by a nonlinear term which turns out to be a challenging feature of the following analysis. First note that we obtain the following equation for the Niwa model:
\begin{corollary}
In strong form the evolution equation of the population density $\rho(\cdot,t)$ at any time $t \in \mathbb{R}_+$, corresponding with~\eqref{eq:CF1_Niwa}, is given by
\begin{align}
\partial_t \rho(s,t) &=
q N\, \int_0^\infty  \Big( \frac{\rho(\hat s,t)  \, \rho(s - \hat s ,t)}{s - \hat s} \chi_{[0,s)}(\hat s)    - \frac{\rho(s,t) \, \rho(\hat s,t)}{\hat s} \Big)  \, d\hat s \nonumber \\
&+ p \Big( s \int_s^\infty \frac{\rho(\hat s, t)}{\hat s^2} d \hat s - \frac{1}{2} \rho(s,t) \Big) . 
\label{eq:jumpproc_strong}
\end{align}
\end{corollary}
\begin{proof}
From \eqref{eq:jumpproc_weakgen} we infer immediately that
\begin{align}
\frac{\rmd}{\rmd t} \int_{\mathbb{R}_+} \varphi(s) \rho(s,t) \rmd s &= \int_{(\mathbb{R}_+)^2} \left( \varphi(s+\hat{s}) - \varphi(s)\right) \rho(s,t) \left( qN \frac{\rho(\hat{s},t)}{\hat{s}} \right)  \rmd s \, \rmd \hat{s} \nonumber \\
&+  \int_{(\mathbb{R}_+)^2} \left( \frac{\hat{s}}{s} \varphi(\hat{s})+ \frac{s-\hat{s}}{s} \varphi(s-\hat{s})- \varphi(s)\right)p \frac{\rho(s,t)}{s} \chi_{[0,s/2]}(\hat{s})\, \rmd s \, \rmd \hat{s}.
\label{eq:jumpproc_weak}
\end{align}
Again, the claim follows by an easy calculation.
\end{proof}
In this case we notice that the merge process is done with rate $q N \rho(\hat s,t)/\hat s $ while the split process is done with rate $p/s$ for every $\hat s \in[0,\frac12 s]$.

\subsection{Derivation of stochastic process} \label{derivatsp}
The key point of our approach is to regard the 
evolution equations~\eqref{eq:jumpproc_stronggen} and~\eqref{eq:jumpproc_strong} 
as master equations for a stochastic process that can be simulated,
if the population density $\rho(\cdot,t)$ is already known.
This stochastic process is essentially of the type whose study
was initiated in McKean's seminal work~\cite{McKean66}.

\subsubsection{Reformulation of deterministic dynamics}
For that purpose, we rewrite the equations as follows.
\begin{lemma}
The evolution law in strong form~\eqref{eq:jumpproc_stronggen} can also be written as
\begin{align} \label{strongform}
\partial_t \rho(s,t) &= \int_{\mathbb{R}_+} \left( K_{\rho(\cdot,t)}^c(\hat{s} \to s)\rho(\hat{s},t) - K_{\rho(\cdot,t)}^c(s \to \hat{s})\rho(s,t) \right) \rmd \hat{s} \nonumber \\
& + \int_{\mathbb{R}_+} \left( K^f(\hat{s} \to s)\rho(\hat{s},t) - K^f(s \to \hat{s})\rho(s,t) \right) \rmd \hat{s},
\end{align}
where 
\begin{equation} \label{rates}
K_{\rho(\cdot,t)}^c(\hat{s} \to s) =  N a(\hat s, s - \hat s)\frac{\rho(s-\hat{s},t)}{s-\hat{s}} \chi_{[0,s)}(\hat{s}), \quad K^f(\hat{s} \to s) =  \frac{s}{\hat{s}} b(s, \hat s -s) \chi_{[s,\infty)}(\hat{s})
\end{equation}
are the coagulation and fragmentation factors. In particular, in the Niwa model these factors are
\begin{equation} \label{Niwarates}
K_{\rho(\cdot,t)}^c(\hat{s} \to s) = q N \frac{\rho(s-\hat{s},t)}{s-\hat{s}} \chi_{[0,s)}(\hat{s}), \quad K^f(\hat{s} \to s) = p \frac{s}{\hat{s}^2} \chi_{[s,\infty)}(\hat{s})\,.
\end{equation}
\end{lemma}
\begin{proof}
The coagulation part follows immediately from \eqref{eq:jumpproc_stronggen}. For the fragmentation part, we observe by using the symmetry between $\hat s$ and $s - \hat s$ that
$$ \frac{1}{s} \int_{0}^s \hat s b(\hat s, s - \hat s) \, \rmd \hat s = \frac{1}{2} \int_{0}^s \left( \frac{\hat s}{s} + \frac{s - \hat s}{s} \right) b(\hat s, s - \hat s) \rmd \hat s = \frac{1}{2} \int_{0}^s  b(\hat s, s - \hat s) \rmd \hat s = \int_{0}^{s/2}  b(\hat s, s - \hat s) \rmd \hat s \,.$$
The factors for the Niwa model follow immediately from~\eqref{DegondrateC}.
\end{proof}
We can summarize these terms, using that $N\rho(s,t)=sf(s,t)$ from \eqref{PopDistr}, into
\begin{align*} 
K_{\rho(\cdot,t)}(\hat{s} \to s)  &= K_{\rho(\cdot,t)}^c(\hat{s} \to s)  + K^f(\hat{s} \to s), \qquad
\lambda_{\rho(\cdot,t)}(s)=
\lambda_{\rho(\cdot,t)}^c(s)+ 
\lambda^f(s) \,, 
\\ 
\lambda_{\rho(\cdot,t)}^c(s) &= 
\int_{\mathbb{R}_+} K_{\rho(\cdot,t)}^c(s \to \hat{s}) \, \rmd \hat{s}
 = \int_0^\infty a(s,r)f(r,t)\,\rmd r \,, 
\nonumber\\ 
\quad \lambda^f(s) &= \int_{\mathbb{R}_+} K^f(s \to \hat{s}) \, \rmd \hat{s}
 = \int_0^s \frac{r}s b(r,s-r)\,\rmd r \,.
\nonumber
\end{align*}
Here, $\lambda_{\rho(\cdot,t)}(s)$ is the rate of change from $s$ to anything else, 
and $\lambda^f(s)$, $\lambda_{\rho(\cdot,t)}^c(s)$ have an analogous interpretation specified to fragmentation and coagulation respectively.
Eq.~\eqref{strongform} can then be written as
\begin{equation} \label{forwardequation}
\partial_t \rho(s,t) =   \int_{\mathbb{R}_+} \lambda_{\rho(\cdot,t)}(\hat s) 
\,\hat\mu_t(\hat s \to s)\,
\rho(\hat{s},t) \, \rmd \hat{s} - \lambda_{\rho(\cdot,t)}(s)\, \rho(s,t)\,,
\end{equation}
where
\begin{equation} \label{transition_prob1}
 \hat\mu_t(\hat s\to s) := \frac{K_{\rho(\cdot,t)} (\hat s \to s ) }{\lambda_{\rho(\cdot,t)}(\hat s)}
\end{equation}
is the corresponding probability of change from $s $ to some fixed $\hat{s}$. 
Formula~\eqref{forwardequation} recalls the classical form of the forward equation for an associated jump process, as for example outlined in \cite[Section X.3]{Feller} or \cite[Chapter 4.2]{EK}. However, the transition rates here depend on the density $\rho(\cdot, t)$ itself, and this makes the equation nonlinear. 

Observe that formulas~\eqref{forwardequation} and thereby~\eqref{strongform} are consistent with the assumption of mass conservation, easily derived as follows:
\begin{align*}
\frac{\rmd}{\rmd t} \int_{\mathbb{R}_+} \rho(s,t) \, \rmd s &=   \int_{\mathbb{R}_+} \rho(\hat{s},t) \int_{\mathbb{R}_+} K_{\rho(\cdot,t)}(\hat{s} \to s) \, \rmd s  \, \rmd \hat{s} -\int_{\mathbb{R}_+}  \lambda_{\rho(\cdot,t)}(s) \rho(s,t) \rmd s \\
&=  \int_{\mathbb{R}_+} \rho(\hat{s},t) \lambda_{\rho(\cdot,t)}(\hat s) \, \rmd \hat{s} -\int_{\mathbb{R}_+}  \lambda_{\rho(\cdot,t)}(s) \rho(s,t) \rmd s = 0\,.
\end{align*}

In the Niwa model,
\begin{equation} \label{Niwalffinite}
 \lambda^f(s) = p \int_{0}^s \frac{1}{s^2}\hat{s} \, \rmd \hat{s} = \frac{p}{2},
\qquad 
\lambda_{\rho(\cdot,t)}^c(s) 
= q\int_0^\infty f(\hat s,t)\,\rmd\hat{s} = q\,m_0(t),
\end{equation}
where $m_0(t)$ is the zeroth moment of the group size distribution 
$f(\cdot,t)$.  
Hence, the rate $\lambda_{\rho(\cdot,t)}$ is finite as long as
$f(\cdot, t) \in L^1((0, \infty))$. 
Generally, according to \cite[Theorem 6.1]{DLP}, there exists a unique global in time solution to~\eqref{eq:CF1_Niwa} in terms of finite non-negative measures on $(0, \infty)$ for any finite non-negative initial measure. In particular, 
the solution was proved to have a smooth density $f(\cdot,t)$ if the initial group-size distribution has a density $f_{\textnormal{in}}$ that is  
completely monotone. 

\begin{proposition}
\label{p:Niwaexist}
\cite[Theorem~6.1]{DLP}
In the Niwa model, i.e.,~with rates~\eqref{Niwarates}, equation~\eqref{forwardequation} has a unique global (in time) solution $ \rho$ such that $ f(\cdot,t) = N \frac{\rho(\cdot,t)}{\cdot} $ is completely monotone, if $ f(\cdot,0) = f_{\rm in} $ is completely monotone with finite zeroth and first moments.  In this case, the 
zeroth moment $m_0(t)$ in \eqref{Niwalffinite} satisfies 
\[
m_0'(t)=\frac12(p\, m_0- q \,m_0^2),
\]
and remains bounded for all $t>0$.
\end{proposition}
Further, we recall that in \cite{DLP} the existence of a unique scaling profile $f_{*}$ for the equilibrium $f_{\textnormal{eq}}$ of~\eqref{eq:CF1_Niwa}, depending on $p$, $q$ and $N$, is proven. The profile $f_*$ is completely monotone, 
with exponential decay as $s \to \infty$, and $f_*(s) = \mathcal{O}(s^{- 2/3})$ as $ s \to 0$, so $f_*\in L^1((0,\infty))$. 
Hence, in the case of the Niwa model~\eqref{eq:jumpproc_strong} we immediately obtain an explicit formula for equilibrium solutions $\rho_{\textnormal{eq}}$ of \eqref{forwardequation} which we will use for our numerical studies later: for $p=q=2$ (to which all parameter choices can be reduced), we conclude from \cite[Theorem 5.1]{DLP} via~\eqref{PopDistr} that the unique equilibrium profile $\rho_*(x) = x f_*(x)$ satisfies
\begin{equation} \label{popequprof}
 \rho_{*}(x) = \tilde \gamma_*(x) e^{-\frac{4}{27}x} \,,
\end{equation}
where $\tilde \gamma_*$ is a smooth function with
$$ \tilde \gamma_*(x)  \sim \frac{x^{1/3}}{\Gamma(1/3)}\,, \text{ when } x \to 0, \quad \tilde \gamma_*(x)  \sim \frac{9}{8}\frac{x^{-1/2}}{\Gamma(1/2)}\,, \text{ when } x \to \infty \,.$$
For mass $N > 0$, the equilibrium density $\rho_{\textnormal{eq}}$ is given by
\begin{equation}
\label{e:rhoeq}
\rho_{\textnormal{eq}} (x) = \frac{x f_{\textnormal{eq}}(x)}{\int_0^{\infty} x f_{\textnormal{eq}}(x) \, \rmd x } = \frac{\frac{x}{N} f_{*}\left( \frac{x}{N}\right)}{\int_0^{\infty} \frac{x}{N} f_{*}\left( \frac{x}{N}\right)\, \rmd x }  
= \frac{\frac{x}{N} f_{*}\left( \frac{x}{N}\right)}{N \int_0^{\infty} y f_{*}\left(y\right) \, \rmd y} = \frac{1}{N} \rho_*\left(\frac{x}{N} \right)\,.
\end{equation}

For general coagulation and fragmentation rates, one has to be careful in order to make sure that $\lambda_{\rho(\cdot,t)}$ is finite, either by restricting the class of admissible $\rho$ or by truncating the domain to some compact set $ E \subset (0, \infty)$.

\subsubsection{Jump process} \label{constr_jumpprocess}

We introduce a jump process $(X_t)_{t \geq 0}$ in the following way: 
Denote the set of probability densities on $(0, \infty)$ by 
\[
\mathcal{P}((0, \infty)):= \{ f \in L^1((0, \infty))\,:\, f \geq 0, \, \left| f \right|_ {L^1} = 1\} .
\]
Assume $\rho: (0, \infty) \times (0, \infty) \to \mathbb{R}$ is Borel-measurable 
with $\rho(\cdot, t) \in \mathcal{P}((0, \infty))$ for all $t \geq 0$. 
For every $t \geq 0$, $s>0$, we 
define a probability measure $\mu_t(s,\cdot)$ by 
setting, for each Borel-measurable subset $ \Gamma \subset (0, \infty)$,
\begin{equation} \label{transition_prob}
 \mu_t(s, \Gamma) := \int_\Gamma \hat\mu_t(s\to\hat s)\,\rmd\hat s = \frac{\int_{\Gamma} K_{\rho(\cdot,t)} (s \to \hat s )\, \rmd \hat s}{\lambda(t,s)}
, \qquad \lambda(t,s) := \lambda_{\rho(\cdot,t)}(s).
\end{equation}
We assume that $\lambda(t,s)$ is uniformly bounded,
and that $\lambda(t,s)$ and $\mu_t(s,\Gamma)$ are continuous in $t$ 
for each $s$ and $\Gamma$.
Then, by classical results of Feller \cite{Feller1940}
(see also \cite[section X.3]{Feller})
there is a unique solution of the backward equation 
\begin{equation} \label{eq:transfunc}
P(r,t,s, \Gamma) = \delta_s(\Gamma) + \int_r^t \lambda (u,s) \int \left(P(u,t,\hat s, \Gamma) - P(u,t,s, \Gamma)  \right) \mu_u(s, \rmd \hat s) \rmd u\,,
\end{equation}
for the transition function of a Markov process.
(For generalizations of Feller's results 
without continuity, see \cite[Lemma 4.7.2]{EK} and \cite{FMS2014}.)
Given any initial distribution $\nu\in{\mathcal P}((0,\infty))$, 
there exists a corresponding Markov (jump) process $(X_t)_{t \geq 0}$ 
with initial density $\nu$ and transition function $P(r,t,s, \Gamma)$
(see \cite[Theorem 4.1.1]{EK} and also \cite[Theorem 8.4]{Kallenberg}). 
This process $X_t$ solves the (time-dependent) martingale problem associated to 
the family of generators $\left(A_{t}\right)_{t \geq 0}$ given by
\begin{equation} \label{generator}
A_{t} f(s) = \lambda(t,s) \int_{\mathbb{R}_+} \left( f(\hat s) - f(s) \right) \mu_{t}(s, \rmd \hat s)\,,
\end{equation}
for all measurable and bounded functions $f\colon(0, \infty) \to \mathbb{R}$.
Moreover, the law of $X_t$ is given by 
\[
\nu_t(\Gamma):=
{\mathbb P}\{X_t\in \Gamma\} = \int_0^\infty \nu(\rmd s) P(0,t,s,\Gamma) .
\]
Due to the assumption that $\lambda(t,s)$ is bounded (see \cite{Feller1940}),
the transition function also satisfies the forward equation 
\begin{equation} \label{e:Pforward}
\frac{\partial P(r,t,s,\Gamma)}{\partial t} = 
\int_{\mathbb{R}_+} \lambda(t,\hat s)\mu_t(\hat s,\Gamma)P(r,t,s,{\rmd}\hat s)
- \int_\Gamma \lambda(t,\hat s)P(r,t,s,\rmd\hat s)
\end{equation}
and consequently, integration against $\nu(\rmd s)$ shows that
the law of $X_t$ also satisfies the forward equation
\begin{equation}\label{e:forwardnu}
\frac{\partial \nu_t(\Gamma)}{\partial t} = 
\int_{\mathbb{R}_+} \lambda(t,\hat s)\mu_t(\hat s,\Gamma)
\nu_t(\rmd\hat s) 
- \int_\Gamma \lambda(t,\hat s)\nu_t(\rmd\hat s).
\end{equation}

In the stationary case when $\rho(s,t)=\rho(s)$ is constant in time, 
the Markov process $(X_t)_{t \geq 0}$ can be constructed
in a standard way \cite[Section 4.2]{EK}, 
from a Markov chain corresponding to 
time-independent rates $\lambda(s) = \lambda_{\rho_{\textnormal{eq}}(\cdot)}(s)$
and transition probabilities $\mu(s, \Gamma) = \mu_t(s, \Gamma)$ from \eqref{transition_prob},
and a sequence of independently and exponentially distributed random variables.

\begin{remark}
Considering Proposition~\ref{p:Niwaexist}, if $\rho$ is taken to be any
solution of the Niwa model corresponding with a completely monotone initial
group-size distribution $f_{\rm in}$,  the boundedness and continuity 
assumptions indeed hold,
and the Markov process $X_t$ is well defined. 
For general coagulation-fragmentation rate kernels $a(s,\hat s)$ and $b(s,\hat s)$,
however, we do not address the technical issue of what hypotheses on the kernels
and on the initial data are sufficient to ensure that the boundedness
and continuity assumptions hold.
\end{remark}

\subsection{Self-consistency}
In order to guarantee self-consistency of this construction, we need to verify that 
if the function $\rho(\cdot,t)$ that is used in the definition of $\mu_t$ and
$\lambda(t, \cdot)$ is additionally assumed to be a solution of \eqref{eq:jumpproc_stronggen}
(and hence \eqref{forwardequation}), 
then the law of the process $X_t$ is given by 
\begin{equation}\label{e:nurho}
\nu_t(\Gamma) = \int_\Gamma \rho(s,t)\,\rmd s.
\end{equation}
That is, we need to show that $X_t$ is distributed according to $\rho(\cdot,t)$ for all $t\geq 0$.  

Let $\hat\nu_t(\Gamma)$ denote the right-hand side of \eqref{e:nurho}.
Then, upon integrating \eqref{forwardequation} over $\Gamma$ and using  
the definitions \eqref{transition_prob}, we find that
\begin{equation}\label{e:forwardnuhat}
\frac{\partial \hat\nu_t(\Gamma)}{\partial t} = 
\int_{\mathbb{R}_+} \lambda(t,\hat s)\mu_t(\hat s,\Gamma)
\hat\nu_t(\rmd\hat s) 
- \int_\Gamma \lambda(t,\hat s)\hat\nu_t(\rmd\hat s).
\end{equation}
We have that $\hat\nu_0=\nu_0= \nu$. Thus to infer $\hat\nu_t=\nu_t$ for all
$t\ge0$ we need the initial-value problem for the forward equation
\eqref{e:forwardnu} to have a unique solution with the properties enjoyed by
both $\nu_t$ and $\hat\nu_t$. 

That the requisite uniqueness holds is ultimately a consequence of 
our assumption on the boundedness of the transition rates $\lambda(t,s)$.
This assumption ensures that the solution $P$ of \eqref{e:Pforward} 
is a conservative transition function, by which we mean that 
$P(r,t,s,{\mathbb R}_+)=1$ for all $0\le r\le t$ and $s>0$.

\begin{proposition}
Let $\rho(s,t)$, $\lambda(t,s)$ and $\mu_t(s,\cdot)$ satisfy the hypotheses
stated in the previous subsection. 
Assume $\rho(\cdot,t)$ is a solution of \eqref{eq:jumpproc_stronggen}
and $\rho(\cdot,0)=\nu$.
Then $X_t$ is distributed according to $\rho(\cdot,t)$
for all $t\ge0$.
\end{proposition}

\begin{proof}
We recall that the proof that $P(r,t,s,{\mathbb R}_+)\equiv1$
follows an iteration argument (see \cite[Theorem 1]{Feller1940}) 
which is also sketched in \cite[Appendix to X.3]{Feller} in the time-homogeneous case.
(The same result is established without continuity conditions in \cite[Theorem 4.3]{FMS2014}.)
By repeating the proof 
after integration against $\nu(\rmd s)$, it follows that $\nu_t(\Gamma)$ is the
minimal non-negative solution of \eqref{e:forwardnu} with $\nu_0=\nu$ 
that yields a measure satisfying $0\le\nu_t(\Gamma)\le1$ for each $t\ge0$. 
By consequence, $\nu_t(\Gamma)\le \hat\nu_t(\Gamma)$ for all $t$ and all $\Gamma$. 
Because $P$ is conservative, it follows $\nu_t({\mathbb R}_+)=1$. 
Hence, for all $\Gamma$,
\[
1 \ge \hat\nu_t({\mathbb R}_+) 
= \hat\nu_t(\Gamma) + \hat\nu_t({\mathbb R}_+\setminus\Gamma)
\ge \nu_t(\Gamma) + \nu_t({\mathbb R}_+\setminus\Gamma)
= \nu({\mathbb R}_+) = 1,
\]
and from this it follows $\hat\nu_t(\Gamma)=\nu_t(\Gamma)$.
\end{proof}

\begin{remark} As illustrated by Feller in \cite{Feller1940},
with unbounded jump rates $\lambda(t,s)$ it is possible for the natural solution  
$P(r,t,s,\Gamma)$ of the backward equation \eqref{eq:transfunc}
to fail to conserve total probability.
We remark that it would be interesting to investigate how
this may be related to the phenomena of gelation and shattering
in solutions of the general coagulation-fragmentation equation 
\eqref{weakgeneral}. 
\end{remark}

\section{Approximation of the stochastic process by a numerical scheme} 
\subsection{The algorithmic scheme used in this paper} \label{algorithm}
We approximate the jump process $(X_t)_{t \geq 0}$ defined in Section~\ref{constr_jumpprocess} in and out of equilibrium by the following numerical scheme. 
Recall that kernels $K_{\rho}$ depend on the probability density $\rho$. Hence, if the process is not stationary, i.e.~in equilibrium, we need to estimate $\rho$ at every time step of the numerical scheme. Assuming $\rho = \rho_{\textnormal{eq}}$ where $\rho_{\textnormal{eq}}$ can be computed or at least approximated, we can study the dynamics in equilibrium. We will make use of our knowledge of the equilibrium profile $\rho_{*}$ in the case of the Niwa model. 

In the following we explain the algorithm for the case in which the initial distribution is not the equilibrium distribution, and therefore we have to estimate $\rho(\cdot,t_n)$ at every time $t_n$. It will become clear how the algorithm is conducted for fixed $\rho_{\textnormal{eq}}$.

At the beginning, we fix the following quantities:
\begin{itemize}
\item an initial distribution $\rho_0$ on an interval $(0,L)$ for some (large) $L > 0$,
\item the coagulation coefficient $a(s, \hat s)$ and fragmentation coefficient $b(s, \hat s)$,
\item the total number of individuals, i.e.~the mass $N$,
\item the number of sample individuals/particles $\tilde{N}$ (not to confuse with the total mass),
\item the time step size $\rmd t$,
\item the bin size $h$ for a partition of the domain\,.
\end{itemize}
We simulate the jump process for the individuals $1, \dots, \tilde N$ by computing, for each time step $i$, the entries of the vector $S(i) = (S_1 (i), \dots, S_{\tilde N}(i))$ where the entry $S_k(i)$ equals the group size of the $k$th individual. In other words, $S(i)$ denotes the vector which contains all the obtained cluster sizes at time $(i-1)\rmd t$. The initial vector $S(1)$ is chosen
according to $\rho_0$, for example a uniform distribution. We divide the interval $(0,L)$ into $M:= Lh$ bins of length $h$, denoted by $B_1, \dots, B_M$.

At every time step $i\geq 1$, we proceed as follows:
\begin{enumerate}
\item We estimate the coagulation and fragmentation probabilities for the centres $(b_l)_{l=1, \dots, M}$ of the bins $(B_l)_{l=1, \dots, M}$:
\begin{itemize}
\item We approximate the density $\rho(s,(i-1) \rmd t)$ by
\begin{equation} \label{densityestimate}
\hat{\rho}(s,i) = \sum_{l=1}^M \frac{n_l}{h\tilde{N}} \chi_{B_l}(s),
\end{equation}
where $n_l$ is the number of entries of $S(i)$ in $B_l$.
\item  Now we calculate the following quantities for all bin centres $b_k, b_l$:
\begin{align*}
K_{\hat{\rho}(\cdot, i)} (b_l \to b_k)&= a(b_l, b_k - b_l) N \frac{\hat{\rho} (b_k - b_l, i)}{b_k - b_l} \chi_{[0, b_k]}(b_l) + b(b_k, b_l - b_k) \frac{b_k}{b_l} \chi_{[b_k, \infty)}(b_l)\\
&= K_{\hat{\rho}(\cdot, i)}^c(b_l \to b_k) + K^f(b_l \to b_k), \\
\lambda_{\hat \rho(\cdot,i)}(b_l) &= \sum_{k=1}^M K_{\hat{\rho}(\cdot, i)}^c(b_l \to b_k) h + \sum_{k=1}^M K^f(b_l \to b_k) h \\
&= \lambda_{\hat{\rho}(\cdot, i)}^c(b_l) + \lambda^f(b_l)\,.
\end{align*}
\end{itemize}
\item We decide for each entry $S_k(i)$ of $S(i)$ if a jump happens, and if yes, where the jump goes to, in the following way:
\begin{itemize}
\item For each $S_k(i)$ , we determine the bin $B_{l_k}$ in which it is contained.
Furthermore, we generate a random number $r \in [0,1]$ from the uniform distribution on the unit interval.
\item If $ r > 1- \exp\left(-\lambda_{\hat \rho(\cdot,i)}(b_{l_k}) \rmd t \right)$, nothing happens and $S_k(i)$ stays in the same bin. Otherwise a jump happens.
\item If a jump happens, we generate another random number $r_1 \in [0,1]$ from the uniform distribution:
\begin{itemize}
\item If $r_1 \leq \frac{\lambda_{\hat{\rho}(\cdot, i)}^c(b_{l_k})}{\lambda_{\hat \rho(\cdot,i)}(b_{l_k})}$, coagulation happens: \\
in this case we generate another random number $r_2 \in [0,1]$ from the uniform distribution and calculate for $1 \leq m \leq M$ the sum 
\begin{equation*} 
P(m) := h\sum_{r=1}^m \frac{K_{\hat{\rho}(\cdot, i)}^c(b_{l_k} \to b_r)}{\lambda_{\hat{\rho}(\cdot, i)}^c(b_{l_k})}  
\end{equation*}  
until $P(m^*)> r_2$. Then we set $S_k(i+1) \in  B_{m^*}$.
\item If $r_1 > \frac{\lambda_{\hat{\rho}(\cdot, i)}^c(b_{l_k})}{\lambda_{\hat \rho(\cdot,i)}(b_{l_k})}$, fragmentation happens: \\
in this case we generate another random number $r_3 \in [0,1]$ from the uniform distribution
and calculate the sum  
\begin{equation*}
P(m) := h\sum_{r=1}^m \frac{K^f(b_{l_k} \to b_r)}{\lambda^f(b_{l_k})}  
\end{equation*}
until $P(m^*)> r_3$. Then we set $S_k(i+1) \in  B_{m^*}$.
\end{itemize}
\end{itemize}
\item In this way, we obtain the vector $S(i+1)$, which contains all the cluster sizes corresponding with the $\tilde{N}$ individuals at time $i \rmd t$.  For time $(i+1)\rmd t$, the procedure starts again with the first step.
\end{enumerate}
Approximating the exponential distribution with time discretization step size $\rmd t$, the algorithm induces a Markov chain, simulating $\tilde N$ trajectories of the jump process $(\tilde X_t)_{ t \geq 0}$ corresponding with the generators
\begin{equation} \label{generatorsimulation}
A_{\hat \rho} f(s) = h \sum_{l,k=1}^M \chi_{B_l}(s) \left( f(b_k) - f(b_l) \right) K_{\hat \rho}(b_l, b_k)\,,
\end{equation}
acting on the bounded and measurable functions $f : (0, \infty) \to \mathbb{R}$.
The process $(\tilde X_t)_{ t \geq 0}$ approximates the jump process from Section~\ref{constr_jumpprocess} with generator~\eqref{generator} for $M \to \infty, L \to \infty$, and, for $\rmd t$ small enough, the simulations give accurate results, as we will demonstrate in Section~\ref{nuumval}. 
\begin{remark} \label{rem:equ}
If we assume that the population distribution is in equilibrium $\rho = \rho_{\textnormal{eq}}$ and we can compute or at least approximate $\rho_{\textnormal{eq}}$ with high accuracy, we can use the algorithm above to simulate the jump process by simply replacing $\hat{\rho}(s,i)$ as in~\eqref{densityestimate} by $\rho_{\textnormal{eq}}(s)$ at each time step $i$. In this case, it is also sufficient to only track one individual for analyzing typical paths; that means that the vector $S$ only has one entry. The approximated process $(X_t)_{\geq 0}$ is a Markov process, 
as indicated in subsection~\ref{constr_jumpprocess} above.

\end{remark}
\begin{remark}
Note that we could also adopt the domain for each step by considering $(0, \max_{1 \leq k \leq \tilde N} S_k(i) + mh)$ for some $m \in \mathbb{N}$, instead of $(0,L)$. However, if $L$ is large enough, the effect of such a measure is vanishingly small due to the fast decay in our models and, therefore, not necessary to obtain an accurate scheme.
\end{remark}
\subsection{Comparison to scheme by Eibeck and Wagner}
As mentioned in the Introduction, there is a long history of stochastic particle methods for coagulation (and fragmentation) equations. Since the nonlinearity is contained in the coagulation terms, works on pure coagulation equations are highly relevant for our class of equations.
For pure coagulation, the standard stochastic model, often referred to as a
Marcus-Lushnikov process \cite{Marcus68, Lushnikov78}, describes a Markov jump
process that models the coagulation of clusters of size $s$ and $\hat s$ to form a single cluster of size $s+\hat s$  with rate kernel $a(s,\hat s)$. 
In quite a number of studies (to be brief, we mention only \cite{Jeon,Norris,EW2003}), the empirical measure for the group-size distribution has been related directly to the coagulation part of the Smoluchowski equation~\eqref{weakgeneral}, in the so-called hydrodynamic limit as the number of particles becomes large.

In work more closely related to the present study, Eibeck and Wagner \cite{EW2001} 
developed a different approximation scheme to study the following mass flow equation for $t\geq 0$ and $\varphi$ continuous and compactly supported:
\begin{equation} \label{massflow}
\int_{0}^{\infty} \varphi(s) Q(\rmd s,t) = \int_{0}^{\infty} Q_0(\rmd s) + \int_0^t \int_{(\mathbb{R}_{+})^2} (\varphi(s+\hat{s}) - \varphi(s)) \frac{a(s,\hat{s})}{\hat s} Q(\rmd s,r) Q(\rmd \hat s,r)\rmd r\,.
\end{equation}
This is the weak form of~\eqref{eq:jumpproc_stronggen} (cf. also~\eqref{massintegral}), which we have referred to as evolution equation of the population density, in the case of pure coagulation. 
A solution $Q$ of~\eqref{massflow} is required to be in the set of all continuous paths with values in the set of non-negative Borel measures, i.e.~$Q \in \mathbb{C}([0, \infty), \mathcal{M}(0,\infty))$.

The solution of \eqref{massflow} is approximated 
by a jump process for the empirical measure of an interacting particle system, formalized as a c\`adl\`ag process with values in a subset of $\mathcal{M}(0,\infty)$. 
This jump process models the interaction of clusters of size $s$ and $\hat s$ to result in a pair of clusters having sizes $s+\hat s$ and $\hat s$. 
In this way, the distribution of cluster sizes in the particle system
is used to estimate particle coagulation rates that determine jump rates for a fixed number of particles,
in a way similar to the algorithm described in the previous subsection.

In more detail, suppose that $a(s, \hat s) \leq h(s) h(\hat s)$ for some continuous function $h$ where $\frac{h(s)}{s}$ is non-increasing. For $N \in \mathbb{N}$, $b_N > 0$ and fixed $\beta > 0$, they define the set of measures
\begin{equation} \label{MbetaN}
\mathcal{M}_{\beta}^N = \left\{ p = \frac{1}{N} \sum_{i=1}^N \delta_{s_i} \in \mathcal M((0, \infty)) \,:\, s_i \in (0,b_N], \int_0^{\infty} \frac{h(s)}{s} p(\rmd s) \leq \beta \right\}\,.
\end{equation}
Defining the map
\begin{equation*}
J(p,s,\hat s) =  \begin{cases} p-\frac{1}{N} \delta_s + \frac{1}{N} \delta_{s+\hat s}, \quad & s+\hat s \leq b_N\,,\\
p-\frac{1}{N} \delta_s, \quad & s+\hat s > b_N\,,
\end{cases}
\end{equation*}
they introduce the generator on continuous and bounded functions $\Phi: \mathcal{M}_{\beta}^N \to \mathbb{R}$
\begin{equation}\label{generatormassflowjump}
\mathcal{G}^N \Phi(p) =  \frac{1}{N} \sum_{1=i,j}^N  \left[ \Phi(J(p,s, \hat s)) - \Phi(p) \right] \frac{a(s, \hat s)}{\hat s} \,,
\end{equation}
which is shown to correspond with a jump process $U^N$. Introducing the set
\begin{equation} \label{Mbeta}
\mathcal{M}_{\beta} = \left\{ p \in \mathcal M((0, \infty)) \,:\, \int_0^{\infty} \frac{h(s)}{s} p(\rmd s) \leq \beta \right\} \supset \mathcal{M}_{\beta}^N\,,
\end{equation}
one can view $U^N$ as a c\`adl\`ag process on $\mathcal{M}_{\beta}$, i.e.~$U^N \in \mathbb{D}([0, \infty), \mathcal{M}(0,\infty))$.

Under these assumptions, Eibeck and Wagner \cite{EW2001} prove weak convergence of $U^N$ to the solution $Q$ of~\eqref{massflow} for $U^N_0 \to Q_0 \in \mathcal{M}_{\beta}$, as $N \to \infty, b_N \to \infty$, and in \cite{EW2003} they provide a similar result for the case where fragmentation is added.
In the corresponding algorithm, the coagulation kernel is replaced by the majorant product kernel $h(s)h(\hat s)$ which leads to a simple computation of the exponentially distributed waiting time for the collision and an independent generation of collision partners (cf. also \cite{EW2000} for the use of majorant kernels). When the collision partners $s_i$ and $s_j$ have been chosen according to the probabilities 
$$\frac{h(s_i)}{\sum_{k=1}^N h(x_k)} \quad \text{and} \quad \frac{h(s_j)/s_j}{\sum_{k=1}^N h(x_k)/x_k},$$
the jump happens with acceptance probability 
$$\frac{a(s_i,s_j)}{h(s_i)h(s_j)},$$
and, in this case, $s_i$ is removed and, if $s_i + s_j \leq b_N$, $s_i + s_j$ is added to the points of the empirical measure.

Note that the essential difference between our scheme and such a method concerns the fact that we simulate single trajectories of individuals jumping between groups of different sizes, while the Eibeck-Wagner algorithm simulates the evolution of the population distribution as a whole as represented by the empirical measure. We track individual trajectories on the state space and therefore we can analyze the statistical properties of such trajectories, as we will see in Section~\ref{sec:statistics}.

\section{Numerical simulations} \label{nuumval}
In the following, we are using the algorithm developed in Section~\ref{algorithm} to simulate the population dynamics for the coagulation-fragmentation model~\eqref{weakgeneral}, or~\eqref{eq:jumpproc_stronggen} in terms of the population density, with coagulation rates $a(s, \hat s)$ and fragmentation rates $b(s, \hat s)$. Firstly, we validate the algorithm by working in the Niwa model (constant rates, see~\eqref{DegondrateC}) where the simulation results can be compared to a known equilibrium distribution. Furthermore, we use the algorithm to study the equilibrium and convergence to equilibrium in the cases of random and polynomial rates which demonstrates the flexibilty of our numerical scheme as opposed to previous ones, see~\cite{DE}.

\subsection{Constant coagulation and fragmentation rates} \label{constcoagfrag}
First we work with the Niwa model~\eqref{eq:CF1_Niwa}, or~\eqref{eq:jumpproc_strong} in terms of the population density, and compare our computation with analytical results. We conduct the numerical scheme described in Section~\ref{algorithm} until a certain time $T > 0$, determining $\hat{\rho}(\cdot,T/\rmd t)$ as in~\eqref{densityestimate}. Using the definition of the population distribution~\eqref{PopDistr}, we can determine 
$$\hat{f}(s,T/\rmd t):=  \frac{N \, \hat{\rho}(s,T/\rmd t))}{s}$$
as an approximation of the size distribution $f(\cdot,T)$ and compare the results with the analytic predictions. 

For doing so we choose total mass $N =1$ and $\tilde{p}=\tilde{q}=1$. Recall from~\cite{DLP} that any other combination of parameters can be reduced to this case by rescaling.
According to \cite{DLP}, the equilibrium size distribution $f_{\textnormal{eq}}$ for \eqref{eq:CF1_Niwa} can be expanded as a series in the following way:
\begin{equation} \label{expansion}
f_{\textnormal{eq}}(x) = \frac{x^{-2/3}}{3} \sum_{n=0}^{\infty} \frac{(-1)^n}{\Gamma(\frac{4}{3} - \frac{2}{3}n)} \frac{x^{n/3}}{n!}\,,
\end{equation}
where $\Gamma$ denotes the gamma function.
We denote the partial sums by
\begin{equation} \label{expansion_terms}
f_{K}(x) = \frac{x^{-2/3}}{3} \sum_{n=0}^{K} \frac{(-1)^n}{\Gamma(\frac{4}{3} - \frac{2}{3}n)} \frac{x^{n/3}}{n!}.
\end{equation}

\subsubsection{Jump process out of equilibrium} \label{outofequilib}
We compute $\hat{f}(\cdot, T/\rmd t)$ according to the algorithm introduced above for $L=30$, $T=20$, $\rmd t=0.01$, bin size $h=0.05$ and $\tilde{N}=10000$ sample individuals which are initially distributed according to the uniform distribution $\rho_0$ on $[0,L]$. As for most simulations in the following, the mass is normalized to $N=1$. In Figure~\ref{fig:fixedrates}, we have averaged $\hat{f}(\cdot, T/\rmd t)$ over larger bin sizes $h_1 =1$ in order to obtain a smoother picture and used linear interpolation to create a continuous plot. The figure compares the computed density according to the algorithm with the analytical approximation, using $f_{50}$~\eqref{expansion_terms}. We show the results on a log-log and a semi-log scale, where log denotes the decadic logarithm in the following, unless stated otherwise. We observe that the algorithm produces convergence to a distribution that approximates the analytical expansion very well. For larger sizes, there are small deviations from the equilibrium due to the extremely small number of observations in this part of the domain. Overall, the stochastic method can be seen to be highly accurate.
\begin{figure}[H]
\centering
\begin{subfigure}{.4\textwidth}
  \centering
  \includegraphics[width=1\linewidth]{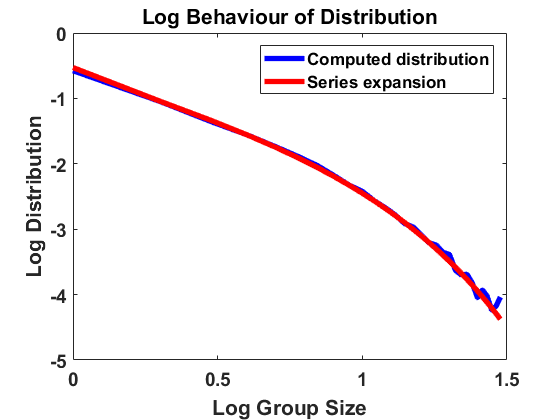}
  \caption{Distributions on log-log scale}
  \label{fig_series1}
\end{subfigure}%
\begin{subfigure}{.4\textwidth}
  \centering
  \includegraphics[width=1\linewidth]{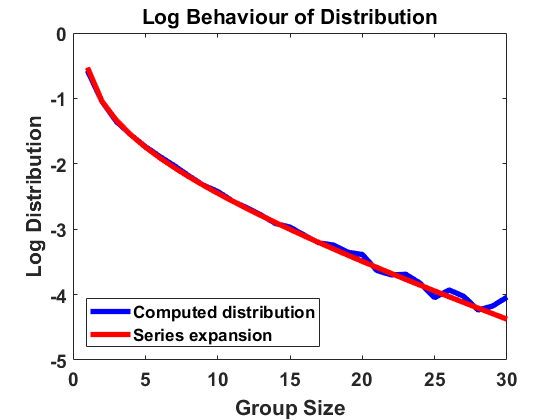}
  \caption{Distributions on semi-log scale}
  \label{fig_series2}
\end{subfigure}
\caption{Equilibrium solution of the Niwa model obtained by the Markov jump process out of equilibrium. For mass $N =1$, bin size $h=0.05$ and $\tilde N=10000$ sample individuals, we compute the size distribution $\hat{f}(\cdot, T/\rmd t)$ up to time $T=20$ following the algorithm introduced in Section~\ref{algorithm} for the Niwa model~\eqref{eq:CF1_Niwa}, starting with a uniform distribution $\rho_0$. We average $\hat{f}(\cdot, T/\rmd t)$ over larger bin sizes $h_1 =1$ and display the interpolated plot together with $f_{50}$ as given in \eqref{expansion_terms}.}
\label{fig:fixedrates}
\end{figure}
\subsubsection{Jump process in equilibrium}
The simulation in equilibrium works with the same algorithm as in Section~\ref{algorithm} but replaces $\hat \rho (\cdot, i)$ at each time step $i$ by the approximation of the equilibrium density $\rho_{\textnormal{eq}}(\cdot)$, obtained from $f_{50}$ as given in \eqref{expansion_terms}, via~\eqref{PopDistr}, see Remark~\ref{rem:equ}. In this situation, 
it is sufficient to approximate the Markov process $X_t$ by simulating the trajectory of a single individual, deploying a simple Monte-Carlo algorithm.

For $L= 30, T=10000$, $\rmd t=0.01$ and bin size $h=0.05$, we approximate the stationary population distribution $\rho_{\textnormal{eq}}$ by measuring time averages of a single trajectory on the interval $(0,L)$ according to the modified algorithm (Remark~\ref{rem:equ}), and, in Figure~\ref{fig:equialg}, we compare the corresponding stationary size distribution with the analytical prediction given by $f_{50}$. We again average over the larger bin size $h_1 =1$ and show continuous plots. As in Figure~\ref{fig:fixedrates}, we show the results on a log-log and a semi-log scale and observe exactly the same as in Section~\ref{outofequilib}. The method is highly accurate except for small deviations from the equilibrium for large group sizes due to the extremely small number of observations in this part of the domain. Hence, we observe that the numerical scheme for simulating the  Markov jump process in equilibrium is consistent with the analytical results.
\begin{figure}[H]
\centering
\begin{subfigure}{.4\textwidth}
  \centering
  \includegraphics[width=1\linewidth]{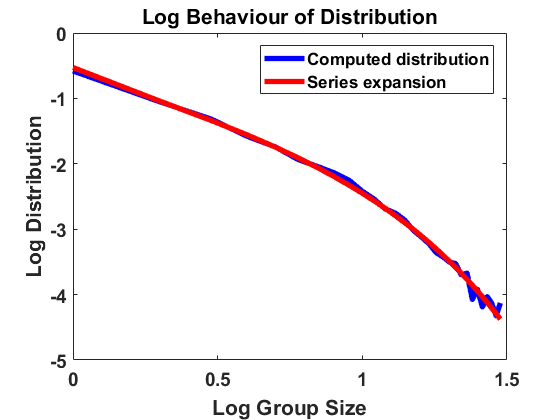}
  \caption{Distributions on log-log scale}
  \label{fig_series1_equialg}
\end{subfigure}%
\begin{subfigure}{.4\textwidth}
  \centering
  \includegraphics[width=1\linewidth]{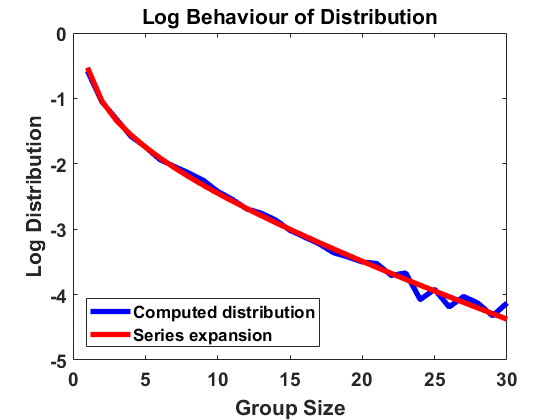}
  \caption{Distributions on semi-log scale}
  \label{fig_series2_equialg}
\end{subfigure}
\caption{Equilibrium solution of the Niwa model obtained by the Markov jump process in equilibrium. For mass $N =1$ and bin size $h=0.05$, we evaluate the time averages of one path in equilibrium up to time $T=10000$ according to the algorithm introduced in Section~\ref{algorithm} (see Remark~\ref{rem:equ}) for the Niwa model~\eqref{eq:CF1_Niwa}, using the equilibrium population density approximated by $f_{50}$ \eqref{expansion_terms} via~\eqref{PopDistr}. As in Fig.~\ref{fig:fixedrates}, we compare the distribution obtained from measuring the time averages on the interval $(0, 30)$ to the analytical approximation of the equilibrium density of size distributions, $f_{\textnormal{eq}}$, again taking $f_{50}$. }
\label{fig:equialg}
\end{figure}

\subsection{Non-constant coagulation and fragmentation rates}
As opposed to the analysis in \cite{DLP} and most numerical methods presented in \cite{DE}, the jump process approach and the associated algorithm do not rely on constant coagulation and fragmentation rates $q$ and $p$. Hence, we use the flexibility of the algorithm to investigate non-constant choices for $a(s, \hat s)$ and $b(s, \hat s)(s+ \hat s)$. Hereby, we test the sensitivity of the results to changes in the model. We observe that random rates produce a clearly different outcome if the variance is high. Similarly, the equilibrium corresponding with polynomial rates separates from the Niwa equilibrium with increasing order of the polynomials. Furthermore, we study random and polynomial variations of the Aizenman Bak model \cite{AizenmanBak} where $a(s, \hat s)$ and $b(s, \hat s)$ are constant and the equilibrium size distribution $f_{\textnormal{eq}}^{AB}$ satsifies the detailed balance condition
$$a(s, \hat s) f_{\textnormal{eq}}^{AB}(s) f_{\textnormal{eq}}^{AB}(\hat s)= b(s, \hat s) f_{\textnormal{eq}}^{AB}(s + \hat s)\,.$$
\subsubsection{Random rates}
In this section, we consider the coagulation and fragmentation rates $a_t(s, \hat s)$ and $b_t(s, \hat s)(s+ \hat s)$ in model~\eqref{weakgeneral}, and thereby~\eqref{strongform}, to be time-dependent. Furthermore, for all $t\geq 0$ and $(s, \hat s) \in \mathbb{R}_+^2$, under preservation of symmetry, they are assumed to be log-normally distributed $\delta$-correlated random variables with mean $q$ or $p$ respectively and standard deviation $\sigma$.  This means that for all $t \geq 0$, $ s, \hat s \in \mathbb{R}_+$ the rates are sampled according to 
\begin{equation} \label{randomrates}
\ln(a_t(s, \hat s)) = \ln(a_t(\hat s,  s)) \sim \mathcal{N}(\ln q,\sigma^2)\,, \quad  \ln(b_t(s, \hat s)(s+ \hat s))= \ln(b_t( \hat s, s)(s+ \hat s)) \sim \mathcal{N}(\ln p,\sigma^2) \,,
\end{equation}
and the correlations are given by
\begin{equation} \label{deltacorr}
\mathbb{E} \left[ \ln(a_{t_1}(s_1, \hat s_1)) \ln(a_{t_2}(s_2, \hat s_2))\right] = \sigma^2 \delta(t_1 - t_2) \, \delta \left( \min\{\left| s_1 - s_2 \right| + \left| \hat s_1 - \hat s_2 \right|, \, \left| s_1 - \hat s_2 \right| + \left| \hat s_1 - s_2 \right|\} \right)\,,
\end{equation}
and analogously for $b_t(s, \hat s)(s+ \hat s)$.
If $\sigma = 0$, the model coincides with the Niwa model~\eqref{eq:CF1_Niwa}.

Setting again $\tilde p = \tilde q = 1$ and thereby $p=q=2$, we compute $\hat{f}_h^N(\cdot, T/\rmd t)$ according to the algorithm introduced in Section~\ref{algorithm}, now sampling $a_t(s, \hat s)$ and $b_t(s, \hat s)(s+ \hat s)$ according to~\eqref{randomrates} independently for every bin and at every time step. For the simulations displayed in Figure~\ref{fig:randomrates}, we have chosen bin size $h=0.05$, $\tilde{N}=20000$ sample individuals and time length $T=20$, starting with a uniform distribution $\rho_0$ and comparing the results for different values of the standard deviation $\sigma$. We use the time step size $\rmd t=0.01$ for $\sigma \leq 3$ and $\rmd t =0.001$ for $\sigma = 5$ to account for possible higher jump rates. As before, we average $\hat{f}_h^{\tilde{N}}(\cdot, T/\rmd t)$ over larger bin sizes $h_1 =1$, use linear interpolation to create a continuous plot and compare the computed distribution to the analytical approximation $f_{50}$ in a log-log scale. 
\begin{figure}[H]
\centering
\begin{subfigure}{.3\textwidth}
  \centering
  \includegraphics[width=1\linewidth]{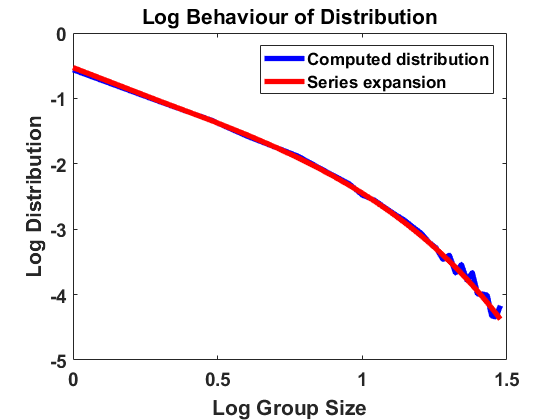}
  \caption{$\sigma =1$}
  \label{fig_series1_random_1}
\end{subfigure}%
\begin{subfigure}{.3\textwidth}
  \centering
  \includegraphics[width=1\linewidth]{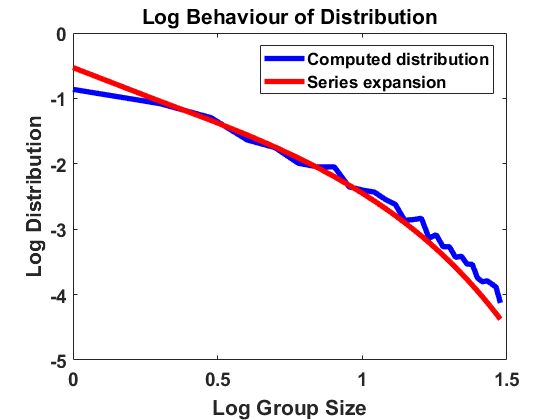}
  \caption{$\sigma =3$}
  \label{fig_series1_random_3}
\end{subfigure}%
\begin{subfigure}{.3\textwidth}
  \centering
  \includegraphics[width=1\linewidth]{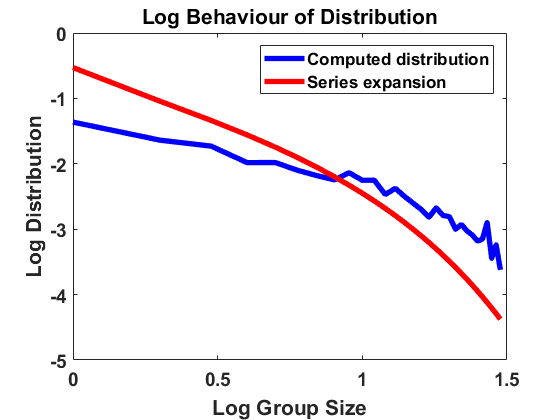}
  \caption{$\sigma =5$}
  \label{fig_series1_random_5}
\end{subfigure}%
\caption{Niwa model with fluctuating coagulation and fragmentation rates. For bin size $h=0.05$, $\tilde{N}=20000$ sample individuals and uniform initial distribution $\rho_0$, we compute the size distribution $\hat{f}(\cdot, T/\rmd t)$ up to time $T=20$ following the algorithm introduced in Section~\ref{algorithm}, with coagulation and fragmentation rates given randomly as in \eqref{randomrates}, for $p=q=2$ and standard deviation $\sigma = 1,3,5$. We average $\hat{f}(\cdot, T/\rmd t)$ over larger bin sizes $h_1 =1$ and display the interpolated plot together with $f_{50}$ as given in \eqref{expansion_terms} in a log-log scale.}
\label{fig:randomrates}
\end{figure}
For $\sigma \leq 1$ the results are almost identical to the ones before, showing almost perfect accordance with the analytical expansion. Hence, the equilibrium appears to be robust under small random fluctuations. However, for $\sigma=3$ we already observe a small discrepancy and for $\sigma = 5$ a clear discrepancy to the model with constant rates. The fact that mass is typically shifted to larger sizes suggests that high random fluctuations favour in average coagulation over fragmentation, even though the fluctuations are equally distributed for both rates. Due to the time-dependent and thereby non-autonomous nature of the random rates, the size distribution cannot reach an equilibrium but a state one could describe as \emph{almost steady}, characterized by small fluctuations around an expected distribution. Making sure such a state is reached here, we have compared the simulation results at $T=20$, $T=30$ and at $T=40$ and observed almost identical behaviour of the size distribution.

Moreover, we consider random fluctuations around the Aizenman-Bak model \cite{AizenmanBak}. This means that for all $t \geq 0$, $ s, \hat s \in \mathbb{R}_+$ the coagulation and fragmentation rates are distributed according to
\begin{equation} \label{randomrates2}
\ln(a_t(s, \hat s)) = \ln(a_t(\hat s, \ s)) \sim \mathcal{N}(\ln q,\sigma^2)\,, \quad  \ln(b_t(s, \hat s))= \ln(b_t( \hat s, s)) \sim \mathcal{N}(\ln p,\sigma^2) \,,
\end{equation}
where the correlations are given as in~\eqref{deltacorr}. Analagously to before, the case $\sigma = 0$ coincides with the Aizenman-Bak model with $a(s, \hat s) = q$ and  $b(s, \hat s) = p$. Setting $p=q=2$, the stationary size distribution is known to be a simple exponential distribution with parameter $1$, i.e.~the stationary density of the size distribution is given by 
$$f_{\textnormal{eq}}^{AB} (s) =e^{-s}.$$

We use the jump algorithm as before to approximate the equilibrium distribution and compare it to the equilibrium of the Niwa model according to the series expansion as well as the equilibrium density $f_{\textnormal{eq}}^{AB} (s) =e^{-s}$ , as shown in Figure~\ref{fig:randomrates2}. First, we set $\sigma = 0$ to compare the distributions without random fluctuations. We observe that the higher fragmentation rates in the Aizenman-Bak model lead to a shift of mass to smaller group sizes for $f_{\textnormal{eq}}^{AB}$ compared with $f_{\textnormal{eq}}$ and that the algorithm approximates $f_{\textnormal{eq}}^{AB}$ very well. Furthermore, we choose $\sigma = 1$ and $\sigma = 5$ to observe that, for small noise, the curves are again similar to the case $\sigma=0$, whereas, for larger noise, the equilibrium for~\eqref{randomrates2} loses mass in the range of small group sizes and gets closer to the equilibrium $f_{\textnormal{eq}}$ \eqref{expansion} for model~\eqref{eq:CF1_Niwa}. In both the Niwa and the Aizenman-Bak model we observe that random rates with large noise drive the size distributions away from the deterministic equilibrium in the direction of a uniform distribution.
\begin{figure}[H]
\centering
\begin{subfigure}{.3\textwidth}
  \centering
  \includegraphics[width=1\linewidth]{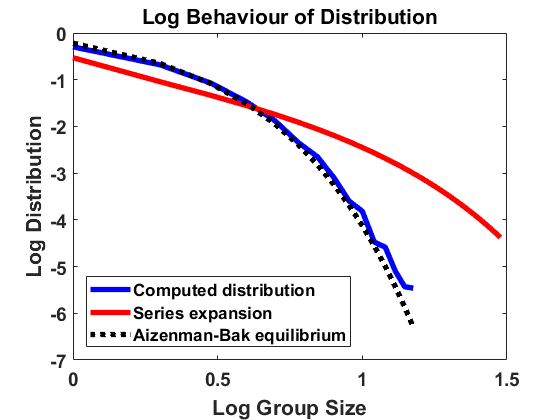}
  \caption{$\sigma = 0$}
  \label{fig_3_random_0}
\end{subfigure}%
\begin{subfigure}{.3\textwidth}
  \centering
  \includegraphics[width=1\linewidth]{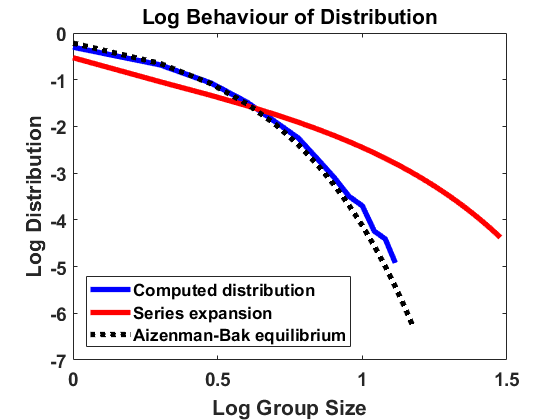}
  \caption{$\sigma = 1$}
  \label{fig_3_random_1}
\end{subfigure}%
\begin{subfigure}{.3\textwidth}
  \centering
  \includegraphics[width=1\linewidth]{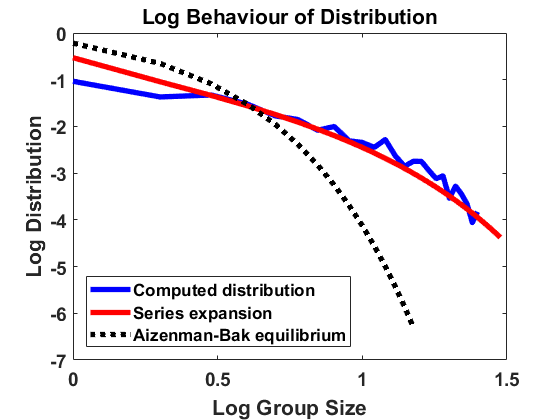}
  \caption{$\sigma = 5$}
  \label{fig_3_random_5}
\end{subfigure}%
\caption{Aizenman-Bak model with fluctuating coagulation and fragmentation rates. Given the same parameters as in Figure~\ref{fig:randomrates}, the results of the analogous simulation, following the algorithm with coagulation and fragmentation rates given randomly as in \eqref{randomrates2} with $p=q=2$, are displayed for $\sigma = 0,1,5$ and compared to the expansion $f_{50}$ from the Niwa model and the stationary density $f_{\textnormal{eq}}^{AB}(s) = e^{-s}$ of the Aizenman-Bak model. The case $\sigma = 0$ corresponds to the Aizenman-Bak model where the equilibrium satisfies the detailed balance condition.}
\label{fig:randomrates2}
\end{figure}

\subsubsection{Polynomial rates}
Furthermore, we consider polynomial coagulation and fragmentation rates. In terms of the physical model describing the dynamics of animal group aggregation, it seems plausible that the fragmentation probability increases with the group size. In addition, larger groups should also have a larger probability to be the result of a coagulation process. The easiest way to implement this reasoning in terms of polynomial rates is given by
\begin{equation} \label{polynomialrates}
a(s, \hat s) = q(s + \hat s)^{\alpha}\,, \quad  b(s, \hat s) = p(s + \hat s)^{\beta -1} \,,
\end{equation}
where $\alpha, \beta \geq 0$, such that for $\alpha =\beta = 0$ we obtain the Niwa model with $a(s, \hat s) = q, \,  b(s, \hat s)(s + \hat s) = p$.
Inserting the rates from~\eqref{polynomialrates} into~\eqref{rates} gives
\begin{equation*}
K_{\rho(\cdot,t)}^c(\hat{s} \to s) =  N q s^{\alpha} \frac{\rho(s-\hat{s},t)}{s-\hat{s}} \chi_{[0,s)}(\hat{s}), \quad K^f(\hat{s} \to s) =  p s \hat{s}^{\beta-2} \chi_{[s,\infty)}(\hat{s})\,.
\end{equation*}
We use the jump algorithm from Section~\ref{algorithm} with rates~\eqref{polynomialrates} to approximate the equilibrium and compare the result to the equilibrium of the Niwa model, estimated by the series expansion. We set $p=q=2$, choose $\alpha=\beta = 0.1,1,3$ and use time step size $\rmd t = 0.01$ for $\alpha = \beta \leq 1$ and $\rmd t = 0.0001$ for $\alpha = \beta = 3$ to account for the higher jump rates. We observe in Figure~\ref{fig:polyrates} that with increasing $\alpha = \beta$ the distribution separates from the equilibrium profile with constant rates. While for $\alpha = \beta = 0.1$ the computed distribution coincides with $f_{\textnormal{eq}}$ and for $\alpha= \beta = 1$ the computed distribution is still very close to $f_{\textnormal{eq}}$, we note that for $\alpha=\beta = 3$ the group sizes are closer to a uniform distribution. Similarly to the situation with random rates, we have compared the simulation results at $T=20$ and at $T=40$ and observed the same behaviour of the size distribution, making sure that the stronger vicinity to the uniform distribution is not caused by a slower convergence process.

The finding that increasing $\alpha=\beta$ imply a divergence from the equilibrium profile with constant rates towards a uniform distribution can be accounted for by the fact that the rate of coagulation to large sizes $s$ increases with $\alpha$. This effect is apparently disproportionate to the impact of an increased rate of fragmentation from large sizes $s$ for increasing $\beta$. 
\begin{figure}[H]
\centering
\begin{subfigure}{.3\textwidth}
  \centering
  \includegraphics[width=1\linewidth]{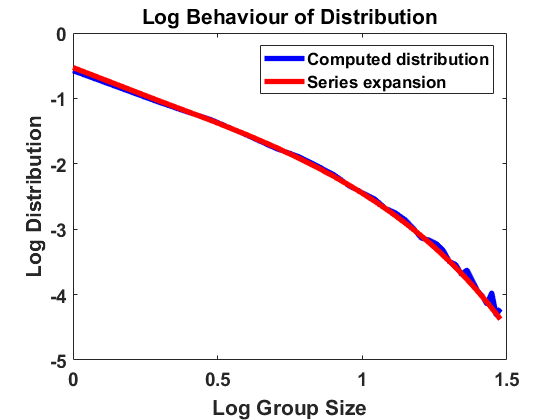}
  \caption{$\alpha = \beta =0.1$}
  \label{fig_polyn_01}
\end{subfigure}%
\begin{subfigure}{.3\textwidth}
  \centering
  \includegraphics[width=1\linewidth]{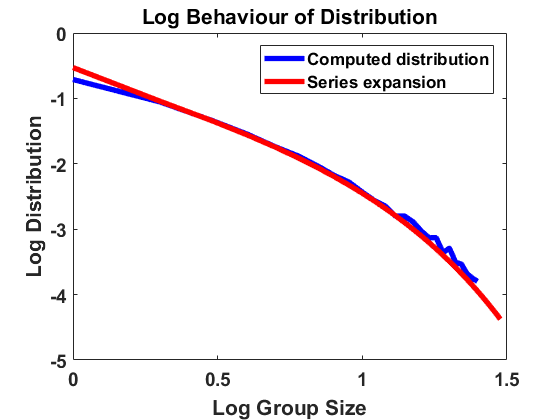}
  \caption{$\alpha = \beta = 1$}
  \label{fig_polyn_1}
\end{subfigure}%
\begin{subfigure}{.3\textwidth}
  \centering
  \includegraphics[width=1\linewidth]{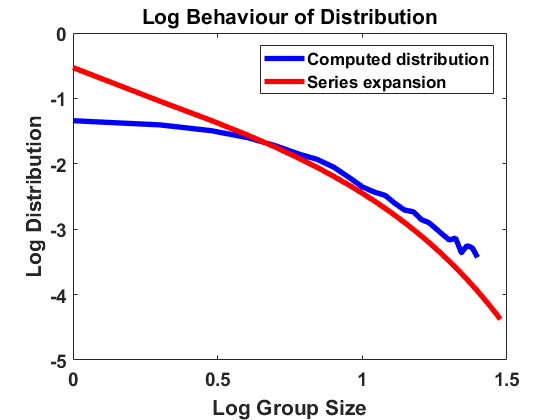}
  \caption{$\alpha = \beta =3$}
  \label{fig_polyn_3}
\end{subfigure}%
\caption{Niwa model with polynomial coagulation and fragmentation rates. For bin size $h=0.05$, $\tilde{N}=20000$ sample individuals and uniform initial distribution $\rho_0$, we compute the size distribution $\hat{f}(\cdot, T/\rmd t)$ up to time $T=20$ following the algorithm introduced in Section~\ref{algorithm}, with coagulation and fragmentation rates given as in \eqref{polynomialrates}, for $p=q=2$ and $\alpha = \beta = 0.1, 1, 3$. We average $\hat{f}(\cdot, T/\rmd t)$ over larger bin sizes $h_1 =1$ and display the interpolated plot together with $f_{50}$ as given in \eqref{expansion_terms} in a log-log scale.}
\label{fig:polyrates}
\end{figure}
Similarly to the previous section, we additionally consider the coagulation and fragmentation rates
\begin{equation} \label{polynomialrates2}
a(s, \hat s) = q(s + \hat s)^{\alpha}\,, \quad  b(s, \hat s) = p(s + \hat s)^{\beta} \,,
\end{equation}
where $\alpha, \beta \geq 0$. In this case, the situation for $\alpha =\beta = 0$ coincides with the Aizenman-Bak model \cite{AizenmanBak} where $ b(s, \hat s)= p$.

As before, for $p=q=2$, we use the jump algorithm to simulate the dynamics with rates~\eqref{polynomialrates2} and approximate the stationary density which we compare to the equilibrium of the Aizenman-Bak model $f_{\textnormal{eq}}^{AB}(s) = e^{-s}$ and the equilibrium of the Niwa model. Recall from Figure~\ref{fig:randomrates2} that the higher fragmentation rates in the Aizenman-Bak model lead to a shift of mass to smaller group sizes, compared with the Niwa equilibrium. We choose the same parameter values as for Figure~\ref{fig:polyrates} and observe in Figure~\ref{fig:polyrates2} that, similarly to the model with rates~\eqref{polynomialrates} as shown in Figure~\ref{fig:polyrates}, the equilibrium distribution seems to be driven away towards a uniform distribution under sufficiently increased exponents $\alpha=\beta$. Similarly to random rates with large variance, the increased coagulation and fragmentation probabilities of large sizes apparently tend to balancing each other out as opposed to the case with constant rates.
\begin{figure}[H]
\centering
\begin{subfigure}{.3\textwidth}
  \centering
  \includegraphics[width=1\linewidth]{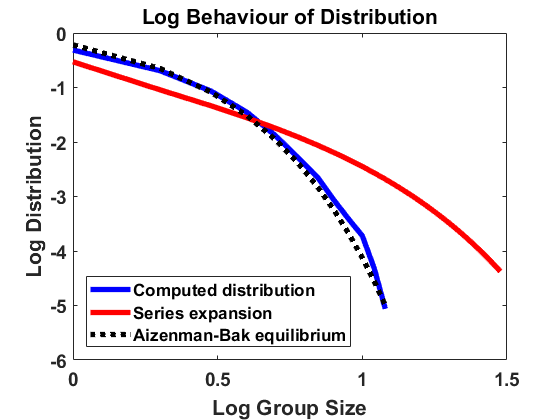}
  \caption{$\alpha = \beta =0.1$}
  \label{fig_polyn_05}
\end{subfigure}%
\begin{subfigure}{.3\textwidth}
  \centering
  \includegraphics[width=1\linewidth]{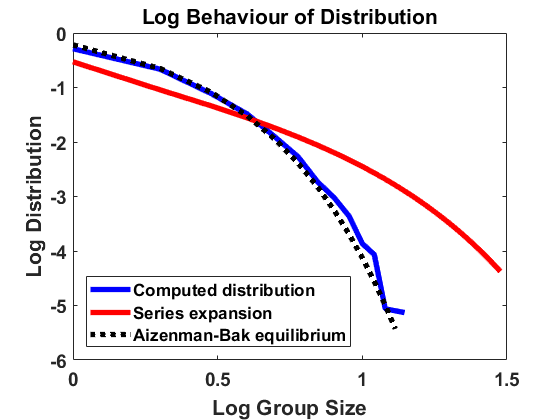}
  \caption{$\alpha = \beta = 1$}
  \label{fig_polyn_1_2}
\end{subfigure}
\begin{subfigure}{.3\textwidth}
  \centering
  \includegraphics[width=1\linewidth]{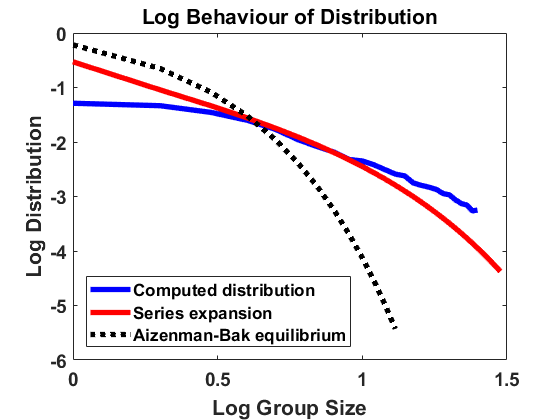}
  \caption{$\alpha = \beta =3$}
  \label{fig_polyn_2}
\end{subfigure}%
\caption{Aizenman-Bak model with polynomial coagulation and fragmentation rates. Given the same parameters as in Figure~\ref{fig:polyrates}, the results of the analogous simulation, following the algorithm with coagulation and fragmentation rates given as in~\eqref{polynomialrates2}, are displayed for $p=q=2$ and $\alpha = \beta = 0,1,3$ and compared to the expansion $f_{50}$ from the Niwa model and the stationary density $f_{\textnormal{eq}}^{AB}(s) = e^{-s}$ of the Aizenman-Bak model.}
\label{fig:polyrates2}
\end{figure}

\section{Statistical analysis of the jump process} \label{sec:statistics}
The stochastic algorithm introduced in Section~\ref{algorithm} can be used to study statistical properties of the jump process $ (X_t)_{t \geq 0}$ from Section \ref{constr_jumpprocess}. In the following, we estimate the decay of correlations for the process in equilibrium and starting out of equilibrium. Furthermore, we approximate the typical occupation times of individuals at cluster sizes for the different types of coagulation and fragmentation rates used in the previous chapter.
\subsection{Autocorrelation function}
The autocorrelation function gives an essential characterization of a stochastic process $ (X_t)_{t \geq 0}$, by measuring the amount of memory the process keeps over times $t-s$. In more detail, let $\mu_t := \mathbb{E}[X_t]$ denote the expected value and $\sigma_t^2 := \mathbb{E}[(X_t- \mu_t)^2]$ denote the variance  of the process at time $t \geq 0$. Then the autocorrelation function is given by
$$\tilde{A}(t,s) =  \frac{\mathbb{E}[ (X_t - \mu_t) (X_s - \mu_s)]}{\sigma_t \sigma_s}\,.$$
Fixing a time $t\geq 0$,  we define the autocorrelation function in one time variable by
$$A(\tau) = \tilde{A}(t, t + \tau).$$
We investigate numerically the autocorrelation function of the process $(X_t)_{t \geq 0}$, as in Section \ref{constr_jumpprocess}, for the Niwa model \eqref{eq:jumpproc_strong} with coagulation and fragmentation parameters $\tilde{p}= \tilde{q}=1$, i.e. $p=q=2$.
We use the jump algorithm described in Section~\ref{algorithm} to make a numerical estimate on the behaviour of $A(\tau)$, given that (a) the stationary density $\rho_{\textnormal{eq}}$ is reached, i.e.~$ \rho_0 = \rho_{\textnormal{eq}}$ (see Remark \ref{rem:equ}), and (b) the equilibrium density is not reached yet but starting from a uniform distribution $\rho_0$. In Fig.~\ref{fig:Autocor}, we observe a rapid decrease of the autocorrelation $A(\tau)$, i.e. fast decay of correlations, in both cases. The findings suggest an exponential decay of correlations with a rate close to $0.25$ in equilibrium and a rate close to $0.3$ starting from a uniform distribution.

The numerical results indicate that, in the Niwa model, the size of the group an individual belongs to at a certain point of time is correlated significantly only to the sizes of the groups the individual belonged to in the close past. This is consistent with the underlying assumption that groups of all sizes can be involved in particular coagulations and fragmentations with equal probability.

\begin{figure}[H]
\begin{subfigure}{.35\textwidth}
  \centering
  \includegraphics[width=1\linewidth]{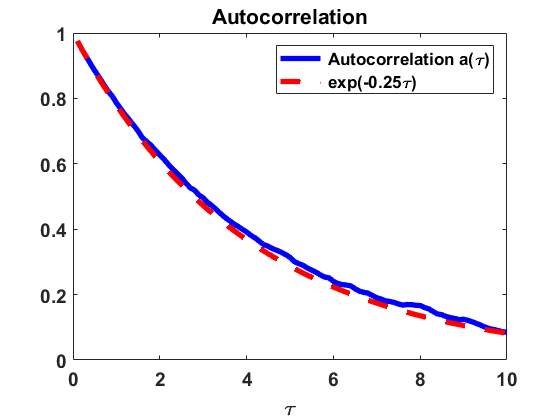}
  \caption{In equilibrium}
  \label{fig_autocor1}
\end{subfigure}
\begin{subfigure}{.35\textwidth}
  \centering
  \includegraphics[width=1\linewidth]{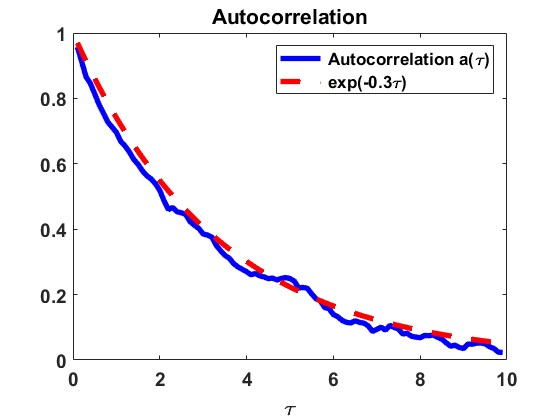}
  \caption{Out of equilibrium}
  \label{fig_autocor2}
\end{subfigure}
\centering
\caption{Autocorrelation functions in the Niwa model. We estimate $A(\tau)$ for $\tau \in (0,10]$ following $10^5$ paths. In (a), the paths are distributed according to the stationary population density $\rho_{\textnormal{eq}}$, evolving according to the algorithm introduced in Section~\ref{algorithm} (Remark \ref{rem:equ}), with $\rho_{\textnormal{eq}}$ estimated by $f_{50}$~\eqref{expansion_terms} via~\eqref{PopDistr}. In (b), the paths are changing the distribution in approach of the equilibrium distribution according to~\eqref{densityestimate}. We observe exponential decay of correlations in both cases.}
\label{fig:Autocor}
\end{figure}

\subsection{Statistics of the occupation time}
The second numerical investigation of the process' statistical properties concerns the time individuals spend in average at a given cluster size before performing a jump to a new cluster size. 
We call this time length the average occupation time of each group size. 
We approximate this magnitude by conducting the algorithm from Section~\ref{algorithm} with the four different types of coagulation and fragmentation rates, presented in Section \ref{nuumval}:
constant rates $\tilde{p}= \tilde{q} =1$ ($p=q=2$), as in the original Niwa model, in equilibrium (Figure~\ref{fig:Average_occup_const}~(a)) and out of equilibrium (Figure~\ref{fig:Average_occup_const}~(b)), random rates as given in~\eqref{randomrates} (Figure~\ref{fig:Average_occup}~(a)-(c)) and  polynomial rates as in~\eqref{polynomialrates} (Figure~\ref{fig:Average_occup}~(d)-(f)). For the random rates we compare $\sigma =1,2,3$ and for the polynomial rates $\alpha= \beta = 1,2,3$.
\begin{figure}[H]
\centering
  \begin{subfigure}{.35\textwidth}
  \centering
  \includegraphics[width=1\linewidth]{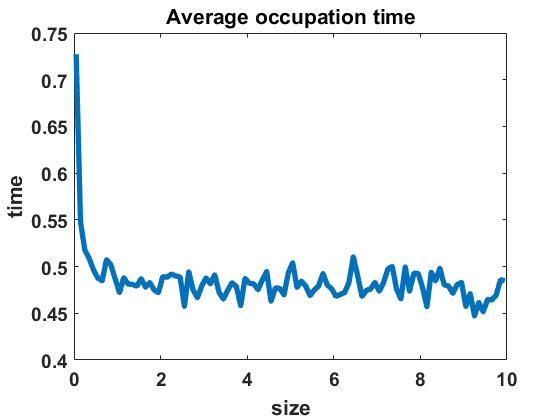}
  \caption{In equilibrium}
  \label{occup_equ}
\end{subfigure}
\begin{subfigure}{.35\textwidth}
  \centering
  \includegraphics[width=1\linewidth]{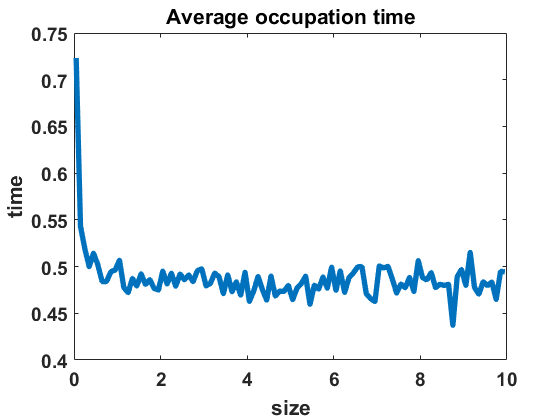}
  \caption{Out of equilibrium}
  \label{occup_nonequ}
\end{subfigure}
\caption{Statistics of average occupation time for the Niwa model in equilibrium (left) and out of equilibrium (right). For $p=q =2$, we deploy the algorithm, introduced in Section~\ref{algorithm}, to simulate $10^3$ paths of the jump process $ (X_t)_{t \geq 0}$ until time $T=10^5$, using a time step size of $ \rmd t = 0.1$. The average occupation time of individuals at each cluster size is measured for coagulation and fragmentation rates~\eqref{DegondrateC}. The density $\rho_{\textnormal{eq}}$ for the equilibrium case (a) is estimated by the expansion~\eqref{expansion_terms} with $K=50$ and the calculations for (b) are conducted with uniform initial distribution.}
\label{fig:Average_occup_const}
\end{figure}
\begin{figure}[H]
\centering
\begin{subfigure}{.3\textwidth}
  \centering
  \includegraphics[width=1\linewidth]{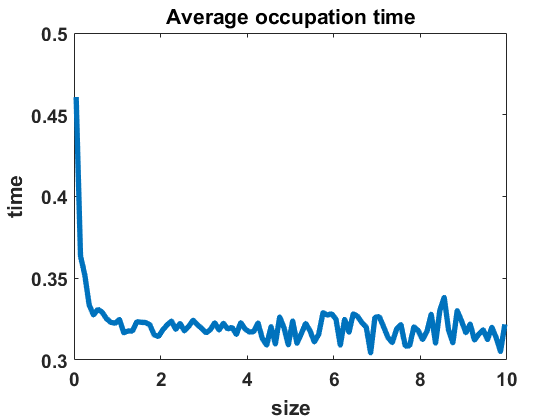}
  \caption{$\sigma = 1$}
  \label{occup_random1}
\end{subfigure}%
\begin{subfigure}{.3\textwidth}
  \centering
  \includegraphics[width=1\linewidth]{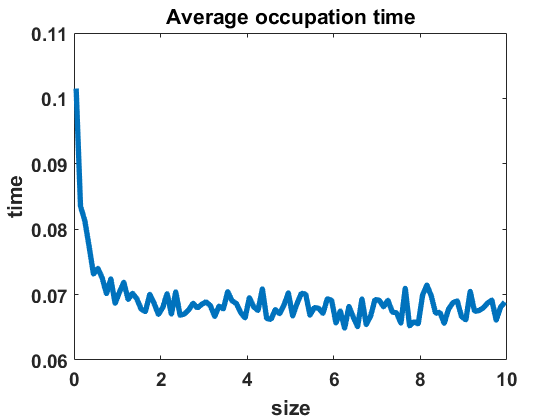}
  \caption{$\sigma =2$}
  \label{occup_random2}
\end{subfigure}%
\begin{subfigure}{.3\textwidth}
  \centering
  \includegraphics[width=1\linewidth]{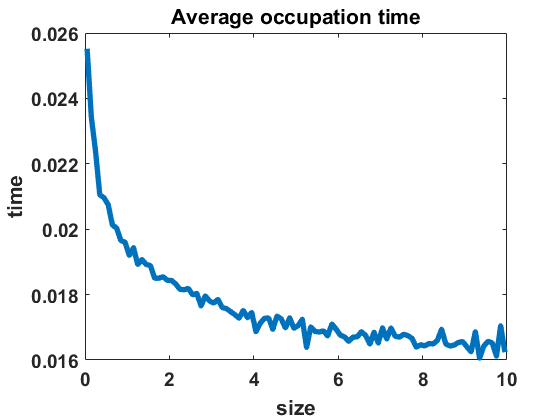}
  \caption{$\sigma = 3$}
  \label{occup_random3}
\end{subfigure}%
\hfill
\begin{subfigure}{.3\textwidth}
  \centering
  \includegraphics[width=1\linewidth]{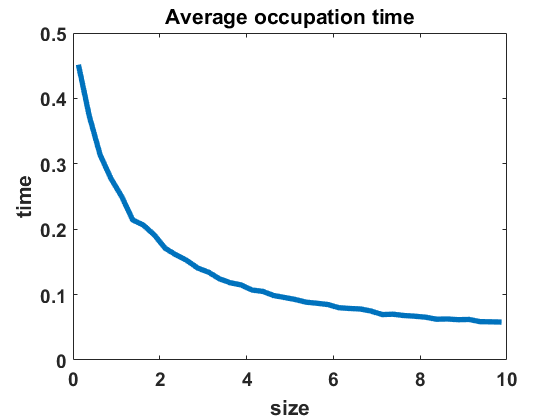}
  \caption{$\alpha = \beta =1$}
  \label{occup_polyn1}
\end{subfigure}%
\begin{subfigure}{.3\textwidth}
  \centering
  \includegraphics[width=1\linewidth]{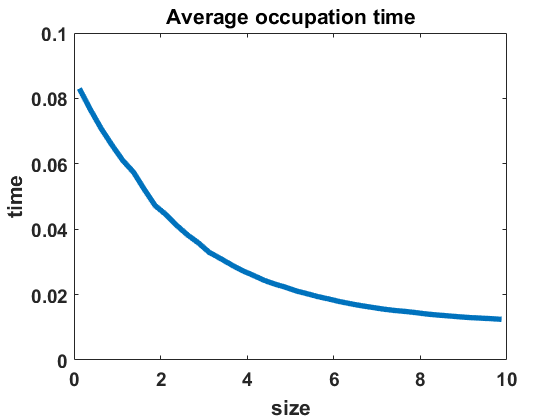}
  \caption{$\alpha = \beta =2$}
  \label{occup_polyn2}
\end{subfigure}%
\begin{subfigure}{.3\textwidth}
  \centering
  \includegraphics[width=1\linewidth]{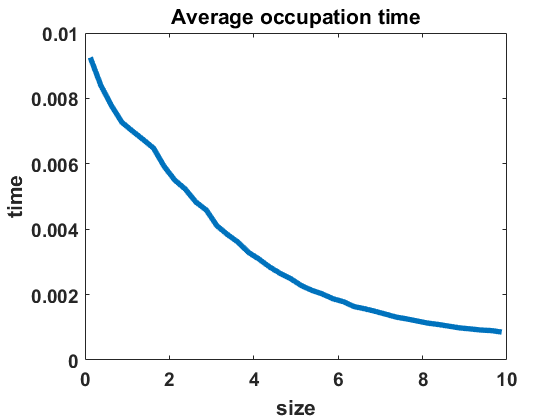}
  \caption{$\alpha = \beta =3$}
  \label{occup_polyn3}
\end{subfigure}
\caption{Statistics of average occupation time for the Niwa model with fluctuating coefficients (top) and with polynomial rates (bottom). Starting from a uniform initial distribution, the same simulations as in Figure~\ref{fig:Average_occup_const} are conducted for random rates~\eqref{randomrates} (with time step size $\rmd t = 0.01$) and polynomial rates~\eqref{polynomialrates} ($ \rmd t = 0.0005$), again measuring the average occupation time of individuals at each cluster size. The average occupation times can be seen to decrease strongly for increasing standard deviation $\sigma$ in the case of random rates ((a)-(c)) and increasing polynomial exponents $\alpha = \beta$ in the situation of polynomial rates ((d)-(f)), due to the increased jump rates.}
\label{fig:Average_occup}
\end{figure}
In Figure~\ref{fig:Average_occup_const} and Figure~\ref{fig:Average_occup}, we display the approximated occupation times, averaging over bin sizes $h_1 = 0.1$, for different choices of coagulation and fragmentation rates. We observe that the occupation times roughly reflect the corresponding equilibrium size distributions (or almost steady size distributions in the random case respectively), as seen in Section~\ref{nuumval}. In the case of coagulation and fragmentation rates~\eqref{DegondrateC} from the Niwa model, as shown in Figure~\ref{fig:Average_occup_const}, the occupation times show a sharp increase for smaller group sizes. We observe similar behaviour for small perturbations of the model, represented by the random case with $\sigma =1$ (Figure~\ref{fig:Average_occup}~(a)) and the polynomial case with $\alpha = \beta =1$ (Figure~\ref{fig:Average_occup}~(d)). Note that for random rates with $\sigma = 1$ and constant rates in equilibrium and out of equilibrium, the decay of occupation times stops at small sizes such that the occupation times fluctuate around a constant level for all larger sizes. In the case of polynomial rates, we observe smooth decay of occupation times mirroring the group size distribution in equilibrium more accurately.

The larger the random or polynomial perturbations of the Niwa model become, the more we see the peak at smaller sizes vanish. In fact, the average occupation times can be seen to decrease to a much smaller scale for increasing standard deviation $\sigma$ in the case of random rates (Figure~\ref{fig:Average_occup}~(a)-(c)) and increasing polynomial exponents $\alpha = \beta$ in the situation of polynomial rates (Figure~\ref{fig:Average_occup}~(d)-(f)), which can be explained by the increased coagulation and fragmentation rates and, thereby, increased jump rates. The decrease of occupation times to a similarly low level for all sizes is in accordance with the equilibrium density tending to a more uniform distribution for increasing random rates and increasing polynomial exponents, as seen in Figures~\ref{fig:randomrates} and~\ref{fig:polyrates}.

\section{SDE approximation to the jump-process model} \label{sec:SDEs}

Motivated by Niwa's approach to use a stochastic differential equation (SDE)
for modeling the dynamics and equilibrium distribution of group sizes \cite{N}, 
we discuss the role of SDEs for modeling the merging-splitting dynamics that
correspond to the coagulation-fragmentation equation~\eqref{weakgeneral} 
or~\eqref{eq:jumpproc_stronggen}. 
First, we discuss Niwa's SDE model and describe some of its problematic aspects. 
Next we derive a natural diffusion approximation to the 
group-size jump process of section \ref{derivatsp},
and demonstrate the inconsistency of this approach for modeling the jump process.
Finally, we discuss an alternative SDE model for the dynamics of group 
sizes --- a stochastic logistic equation --- which, 
although it involves very different mechanisms for group-size changes,
yields equilibrium group-size distributions that also have
the form of a power-law with an exponential cutoff (gamma distribution).

\subsection{Niwa's SDE model} \label{sec:NiwaSDE}
In order to find an expression for the equilibrium group-size distribution,
Niwa \cite{N} models the process $(X_t)_{t \geq 0}$ 
of the size of the group containing a given individual 
via an SDE having the form
\begin{equation} \label{NiwaSDE1}
\rmd X_t = - \frac{\tilde p}{2}(X_t - \bar x) \, \rmd t + \sigma(X_t)\, \rmd W_t\,,
\end{equation}
when $X_t>0$,
where the parameter $\tilde p$ is related to the rate of group splitting per time step,
the constant $\bar x=\ip{X_t}_p$ represents the population-weighted mean group size,
and $W_t$ denotes standard Brownian motion. 
The drift is chosen linearly around the average, which roughly models
the notion that, on average, 
fragmentation decreases group size by half, 
while coagulation increases it by a constant.
Niwa modeled the noise coefficient $\sigma(X_t)$ using data from site-based 
merging-splitting simulations, coming to the conclusion that 
\begin{equation}\label{Niwa-noise}
\sigma(x)^2 = 2D \exp \left(\frac{x}{\bar x}\right)\,,
\end{equation}
where $D$ is a constant. 
Ultimately though, there is no rigorous, or even formal, derivation of \eqref{NiwaSDE1}.

 The stationary Fokker-Planck equation associated with 
the SDE~\eqref{NiwaSDE1} states that $\rmd J/\rmd x=0$ where $J$ is the probability flux
\begin{equation}\label{d:Jflux}
J(x) = -\frac {\tilde p} {2} (x-\bar x)\rho(x) - 
\frac{\rmd}{\rmd x}\left( D \exp\left(\frac{x}{\bar x}\right) \rho(x)\right) \,.
\end{equation}
Taking $J\equiv 0$ to ensure there is no flux at $\infty$, one can solve this equation
to find that  the equilibrium population distribution takes the form
\begin{equation} \label{Niwaequil1}
\rho(x) = \frac1Z \exp \left[ - \frac{x}{\bar x} 
\left( 1- \gamma e^{-x/{\bar x}} \right) \right]\,,
\quad \gamma = \frac{\tilde p \bar x^2}{2D} \,,
\end{equation}
where $Z>0$ is a normalization constant. 
Correspondingly, the stationary group-size distribution is given as
\begin{equation} \label{Niwaequil}
\Phi(x) =  \frac{\rho(x)}x = \frac1{xZ} \exp \left[ - \frac{x}{\bar x} 
\left( 1-\gamma  e^{-x/\bar x} \right) \right]\,.
\end{equation} 

One problem in using \eqref{NiwaSDE1} to model group-size evolution is
that the process $(X_t)$ will hit $0$, and one needs to specify how 
group size will be kept positive.  Niwa appears to model this using
symmetrization after a change of variables, and is led to impose the condition
\begin{equation}\label{Niwa-fdrel}
D = \tilde p \bar x^2 \,,
\end{equation}
corresponding to $\gamma=\frac12$ 
(apparently in order to make a symmetrized drift potential continuously
differentiable at 0). It seems more natural mathematically, instead,
to simply require the stochastic process $X_t$ to reflect at 0. 
As described in \cite{McKean,Tanaka}, e.g., this means that 
a term $\rmd L_t$ is added to the right-hand side of \eqref{NiwaSDE1}, 
where $L_t$ is the local time of the process $X_t$ at 0, 
determined by the formula
\begin{equation}\label{d:localtime}
L_t = \lim_{\delta\to0^+} \frac1{2\delta} \int_0^t \one_{\{X_s<\delta\}}\,\rmd s\,.
\end{equation}  
The equilibrium density of this reflected process still has the form in
\eqref{Niwaequil1} with $J\equiv 0$, 
with normalization constant $Z$ simply chosen to make
$\rho$ a probability density on $(0,\infty)$.

This leaves $\gamma$ as a free parameter in the model, which
one ought to specify in some further way. 
In terms of the quality of fitting \eqref{Niwaequil} to the empirical data shown in 
\cite[Fig. 5]{N}, it does not matter much what the precise value of $\gamma$ 
is, as long as it is small. On the scale of \cite[Fig. 5]{N}, the value $\gamma=0$ 
provides a very acceptable fit, as was mentioned by Niwa himself
\cite{MJS} and was shown in \cite[Fig. 2]{DLP}.
The simulation data Niwa generated in \cite[Fig. 2]{N} 
seem to be consistent with a much larger value of $\gamma$, however,
that would not lead to a good fit with the data of \cite[Fig. 5]{N}.

Since the stationary distribution \eqref{Niwaequil} reasonably 
fits empirical data, one can consider whether the SDE \eqref{NiwaSDE1}
is a suitable basis for numerical simulation of the individual group-size 
process. We perform simulations 
using a simple Monte-Carlo scheme for an Euler-Maruyama integration  of~\eqref{NiwaSDE1} with reflection and with step size $\rmd t=10^{-4}$, 
taking $\tilde p=1=\bar x$ and imposing \eqref{Niwa-fdrel}.
(This means that the corresponding $\hat f_{\textnormal{eq}}$ with $\bar x=1$, $N=1$ has to be rescaled as $\hat f_{\textnormal{eq}} = 36 f_{\textnormal{eq}} (6x)$, see \cite[Remark 5.1]{DLP}.)   
Following one trajectory up to time length $T=10^6$, we approximate the stationary size distribution~\eqref{Niwaequil} on $(0, \infty)$. 
In Fig.~\ref{fig:NiwaSDE}(a), we compare the result of the simulation, the density~\eqref{Niwaequil} and the rigorously derived equilibrium $\hat f_{\textnormal{eq}}$ with $\bar x = 1$, $N=1$, estimated by the rescaled expansion $f_{50}$.  
We observe that for small group sizes the approximation is relatively close to the other two densities, but for larger group sizes trajectories are lost although the time step size is already extremely small. This has to do with the highly unstable diffusion coefficient which is an exponential function. We actually can't compute the distribution for sizes $ x > 10$ due to the unstable diffusion coefficient.

In order to avoid the exponentially unstable diffusion coefficient, we apply the following change of variables. Recall we take $\tilde p=1=\bar x$, $\gamma=\frac12$. Similarly to Niwa \cite{N}, for $X_t>0$ we introduce $ Y_t \in (0,1)$ by
\begin{equation} \label{changeofvar}
Y_t =  1 - \exp(-X_t/2) \,.
\end{equation}
We formally use It\^{o}'s formula to obtain, for $Y_t>0$,
\begin{equation} \label{changedSDE}
\rmd Y_t  = \frac14 \left[ ( 2 \ln( 1-  Y_t ) + 1 ) (1- Y_t)  
- \frac{1}{ 1- Y_t } \right] \rmd t 
 + \frac{1}{\sqrt{2}} \rmd W_t\,.
\end{equation}
Again we require the process $Y_t$ to be reflected at 0, so this equation should be modified by a local time term.  
The strong negative drift near 1 prevents the exact process $Y_t$ from
hitting 1 (as one can check using the criterion from \cite[Theorem 3]{Feller1954}, see Section~\ref{ss:logistic} below).
In a Monte-Carlo simulation, however, we have to prevent trajectories from leaving the domain at 1, by simply letting them stay at the same position in case the absolute value would become larger than $1$. This Monte-Carlo algorithm, based on~\eqref{changedSDE}, yields Fig.~\ref{fig:NiwaSDE}~(b), where again the result of the simulation (with $\rmd t = 10^{-3}$, $T=10^{6}$) is compared to the stationary density~\eqref{Niwaequil} and the rigorously derived equilibrium $\hat f_{\textnormal{eq}}$ with $\bar x =1$, $N= 1$, estimated by the rescaled expansion $f_{50}$. We observe that the distribution obtained by the simulation lies close to both densities but does not coincide with either of them which can be seen in particular for larger sizes.

\begin{figure}[H]
\centering
\begin{subfigure}{.4\textwidth}
  \centering
  \includegraphics[width=1\linewidth]{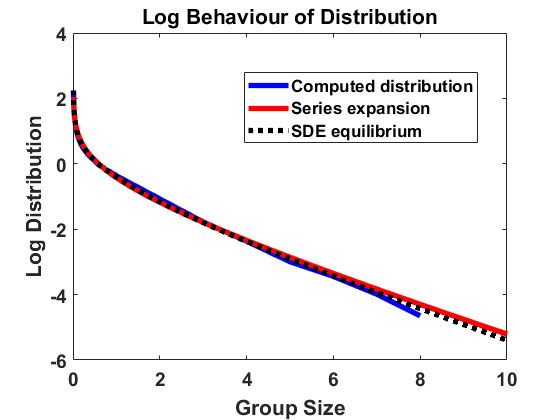}
  \caption{Simulation of~\eqref{NiwaSDE1}}
  \label{fig_NiwaSDE4}
\end{subfigure}%
\begin{subfigure}{.4\textwidth}
  \centering
  \includegraphics[width=1\linewidth]{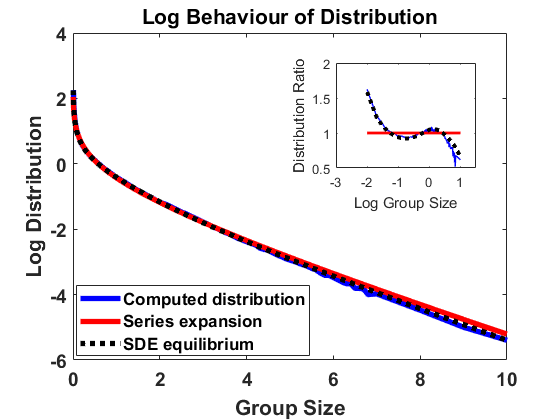}
  \caption{Simulation of~\eqref{changedSDE}}
  \label{fig_Niwa_cov_2}
\end{subfigure}
\caption{ Semi-log plot of group-size distribution of trajectories 
of (a) Niwa's SDE \eqref{NiwaSDE1} and (b) the transformed SDE~\eqref{changedSDE}, 
obtained by Euler-Maruyama integration 
(using time length $T=10^6$ and step size $\rmd t=10^{-4}$ in (a), $\rmd t = 10^{-3}$ in (b)),
compared to the equilibrium density~\eqref{Niwaequil} and the equilibrium $\hat f_{\textnormal{eq}}$ for model~\eqref{eq:CF1_Niwa} with $\bar x = 1$, $N= 1$, estimated by the rescaled expansion $f_{50}$. 
In (b), \emph{inset} log-linear plot of ratio between the respective distributions and $f_{50}$, similarly to \cite[Figure 2]{DLP} where the normailzation factor $1/Z= 0.881237$ for the Niwa SDE equilibrium was not taken into account.}
\label{fig:NiwaSDE}
\end{figure} 

Summarizing, modeling the dynamics via~\eqref{NiwaSDE1} or~\eqref{changedSDE} lacks rigorous justification. The uniformly elliptic noise pushes the process to hit the origin and one must invoke {\it ad hoc} a means to keep it positive,   
unjustified
in terms of the underlying population dynamics as originally outlined by Niwa. 

Even though the simulated processes seem to reach an equilibrium close to the analytical prediction (see Figure~\ref{fig:NiwaSDE}), another serious modeling issue is that SDE sample paths are always continuous in time, and do not make large jumps in the way the merging/splitting mechanism would suggest.  It is not clear whether the solution process of such an SDE can be related to the Markov jump process which is derived in Section~\ref{derivatsp} and simulated successfully in Section~\ref{nuumval}. The next section explores the possibility of such a connection.

\subsection{SDE and the jump process} \label{SDEjump}
In the following, we investigate the suitability of a natural drift-diffusion approximation to 
the jump process constructed in Section~\ref{constr_jumpprocess}, in the situation of the Niwa model with coagulation and fragmentation factors~\eqref{DegondrateC}.
Recall from~\eqref{generator} the family of generators $(A_t)_{t \geq0}$
\[
(A_t f) (x)= \lambda(t,x) \int (f(y)-f(x))\mu_t(x, \rmd y)\,,
\]
where $\mu_t$ and $\lambda(t,x)$ are given by~\eqref{transition_prob}.
Writing, similarly to before, $\mu_t(x, dy) = \hat \mu_t(x,y) d y$, the forward equation for the jump process is the Fokker-Planck equation
\[
\partial_t \rho(x,t) = A_t^*\rho = \int \lambda(t,y) \hat \mu_t(y,x) \rho(y,t)\,\rmd y - \lambda(t,x)\rho(x,t)\,,
\]
as given in \eqref{forwardequation}. Matching \cite{DLP}, as before, we take $N=1$ and $p=q=2$ and assume that we are in equilibrium, i.e.~$\rho(x,t) = \rho_{\textnormal{eq}}(x)$ for all $t \geq0$, such that $\lambda(x)= \lambda_{\rho_{\textnormal{eq}}(\cdot)}(x) = \lambda(t,x)$, $\hat \mu(y,x) = \hat \mu_t(y,x)$ and $A = A_t$ are time-independent.
In equilibrium, we observe that
\[
\lambda(y) \hat \mu(y,x) = K_{\rho_{\textnormal{eq}}(\cdot)}(y\to x)
= 2 f_{\textnormal{eq}}(x-y)\one_{x>y} + 2\frac{x}{y^2}\one_{x<y} \ .
\]
Using the moment relations in \cite[Eq. (5.6)]{DLP} yields
\[
\lambda(y) = \int_0^\infty K_{\rho_{\textnormal{eq}}(\cdot)}(y\to x)\,\rmd x = 2 m_0(f_{\textnormal{eq}})+1 = 3\,,
\]
where $m_k(f):= \int_{\mathbb{R}_+} x^kf(x) \, \rmd x$.
This means constant event rates, which is consistent with the fact that the system is in equilibrium.  We obtain
\begin{equation} \label{muequil}
\hat \mu(y,x) = \frac23\left( f_{\textnormal{eq}}(x-y)\one_{x>y} + \frac{x}{y^2}\one_{x<y}\right) \ .
\end{equation}
Now, supposing that the jumps go typically to a close range of sizes, we use the Taylor approximation 
\[
f(y) = f(x) + f'(x)(y-x) + \frac{f''(x)}{2} (y-x)^2 + o((y-x)^2),
\]
to get
$A f \approx A_{\textnormal{dd}} f$, where the drift-diffusion approximation $A_{\textnormal{dd}}$ to the 
jump-process generator $A$ is given by 
\begin{equation} \label{Taylorgenerator}
(A_{\textnormal{dd}} f)(x) = b(x)f'(x) + \frac12 c(x) f''(x)\,,
\end{equation}
with
\[ 
b(x) = \lambda(x)\int (y-x) \hat \mu(x,y) \, \rmd y,\qquad
c(x) = \lambda(x)\int (y-x)^2 \hat \mu(x,y) \, \rmd y.
\]
Note that $b$ represents the drift and $c$ the diffusion coefficient.
We can conclude from~\eqref{muequil} that in equilibrium these coefficients are as follows: using $m_1(f_{\textnormal{eq}})=N=1$, the drift is given as
\begin{align} \label{driftcoeff}
b(y) &= 2\int_0^\infty(x-y)f_{\textnormal{eq}}(x-y)\one_{x>y}\,\rmd x + \frac2{y^2} \int_0^y (x-y)x\,\rmd x \nonumber
\\ & = 2+ 2y \left(\frac13-\frac12\right) = 2\left(1-\frac y6\right).
\end{align}
Note that the signs are consistent with the model since the drift pushes to the right at small $y$ and the left at large $y$, as expected.
Using that $m_2(f_{\textnormal{eq}})=6$ from \cite[Eq. (5.6)]{DLP}, the diffusion coefficient reads
\begin{align} \label{diffusioncoeff}
c(y)&= 2\int_0^\infty(x-y)^2f_{\textnormal{eq}}(x-y)\one_{x>y}\,\rmd x + \frac2{y^2} \int_0^y (y-x)^2x\,\rmd x \nonumber
\\ & = 12+ 2y^2 \left(\frac13-\frac14\right) = 12 + \frac16 y^2.
\end{align}
For the SDE with drift $b$ and diffusion $c$, the Fokker-Planck equation corresponding with~\eqref{Taylorgenerator} is
\[
\partial_t \rho + (b\rho)_x = \frac12 (c\rho)_{xx}.
\]
The stationary solution $\rho_{\textnormal{dd}}$ of this equation satisfies
\[
\frac{\rho_{\textnormal{dd}}'}{\rho_{\textnormal{dd}}}
= \frac{2b}c -\frac{c'}{c}
= \frac{6(4-x)}{72+x^2} .
\]
After integration, we can determine
\begin{equation}
\rho_{\textnormal{dd}}(x) = 
\rho_{\textnormal{dd}}(0)
\frac{72^3 A(x)}{(72+x^2)^3}\, ,
\qquad A(x) = \exp (2\sqrt 2 \tan^{-1} (x/\sqrt{72})).
\end{equation}

The stationary density $\rho_{\textnormal{dd}}$ for the approximating SDE with generator~\eqref{Taylorgenerator} 
differs rather substantially from 
the stationary solution of~\eqref{eq:jumpproc_strong} for $p=q=2$, $N=1$ which is given by $\rho_{\textnormal{eq}} = \rho_*$~\eqref{popequprof}.  
See the comparison of group size distributions in Figure~\ref{fig_rhodd}, 
and note that each tick mark on the vertical scale corresponds to 2 orders of magnitude.
For small group size, we have $\rho_{\textnormal{dd}}(x)\sim$ const,
while $\rho_*(x)\propto x^{1/3}$ from \eqref{popequprof}. 
Furthermore, while $\rho_*(x)$ decays exponentially, $\rho_{\textnormal{dd}}$ decays only algebraically fast with $\rho_{\textnormal{dd}}(x)\propto x^{-6}$ as $x\to\infty$. 
One trouble is that for large sizes, 
the drift and diffusion rates are dominated by the (uniform) fragmentation mechanism,
which is not well-described by small jumps. 

Recall that Niwa estimates the diffusion coefficient for the SDE~\eqref{NiwaSDE1} by fitting it into a semi-log plot of the variance of size changes in finite time intervals, based on data obtained from site-based simulations of merging and splitting \cite[Figure 2]{N}. 
We can deploy the algorithm for the Markov jump process in equilibrium to approximate the variance of size changes and compare the computations to the diffusion coefficient $c$. In Figure~\ref{fig_var_2b} we observe that size changes exhibited by the simulated jump process differ noticeably from the function $c$ (scaled by $\rmd t=0.05$ as appropriate), 
but not by a large percentage.  By fitting on a semi-log scale as indicated in Figure~\ref{fig_var_3_log}
we find a good fit with a similar exponential form as Niwa had, namely with
\begin{equation}\label{eq:fit_c1}
c_1(y) = \exp(2.19+0.1y).
\end{equation}

\begin{figure}[H]
\centering
\begin{subfigure}{.4\textwidth}
  \centering
  \includegraphics[width=1\linewidth]{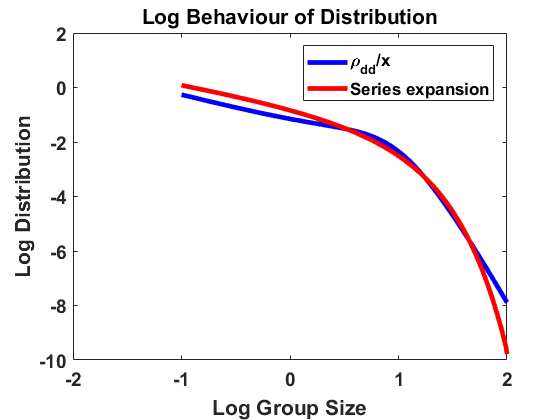}
  \caption{Comparison of $\rho_{\textnormal{dd}}(x)/x$ with $f_{\textnormal{eq}}(x)$ }
  \label{fig_rhodd}
\end{subfigure}%
\begin{subfigure}{.4\textwidth}
  \centering
  \includegraphics[width=1\linewidth]{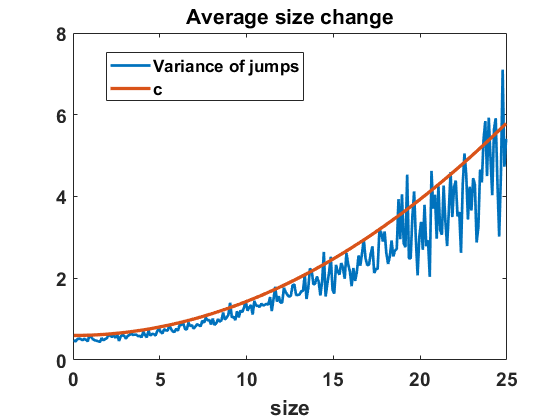}
  \caption{Variance comparison to $c$}
  \label{fig_var_2b}
\end{subfigure}
\caption{Validity of the SDE approximation with drift $b$ \eqref{driftcoeff} and diffusion coefficient $c$ \eqref{diffusioncoeff}. For mass $N =1$ and $p=q=2$: (a) We compare the equilibrium group size distributions $f_{\textnormal{dd}}(x)=\rho_{\textnormal{dd}}(x)/x$ and $f_{\textnormal{eq}}$; (b) We simulate the jump process in equilibrium (see Section~\ref{algorithm}) and estimate the variance of size change by averaging along trajectories with time increment $\rmd t=0.05$. We compare the computations to the diffusion coefficient $c$ (scaled by $\rmd t$), as calculated from the second order approximation~\eqref{diffusioncoeff}.}
\label{fig:variance}
\end{figure}

\begin{figure}[H]
\centering
\begin{subfigure}{.4\textwidth}
  \centering
  \includegraphics[width=1\linewidth]{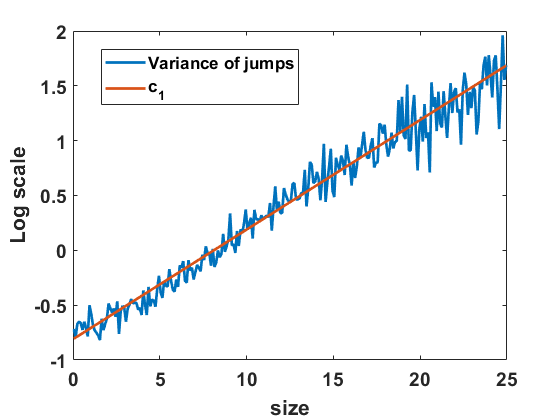}
  \caption{Variance comparison to $c_1$, semi-log plot}
  \label{fig_var_3_log}
\end{subfigure}
\begin{subfigure}{.4\textwidth}
  \centering
  \includegraphics[width=1\linewidth]{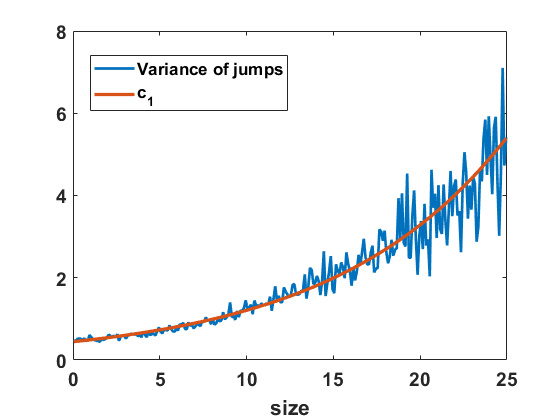}
  \caption{Variance comparison to $c_1$, linear plot}
  \label{fig_var_3}
\end{subfigure}%
\caption{Fitting of exponential diffusion coefficient \eqref{eq:fit_c1}. For mass $N =1$ and $p=q=2$, we simulate the jump process in equilibrium (see Section~\ref{algorithm}) and estimate the variance of size change by averaging along trajectories. We compare the computations to the fitted coefficient $c_1$
in \eqref{eq:fit_c1}.}
\label{fig:variance_c1}
\end{figure}

\subsection{Model with degenerate noise}
So far, we have ascertained that the adequate stochastic method for studying the coagulation-fragmentation model~\eqref{weakgeneral} with non-local rates in general and the Niwa model~\eqref{eq:CF1_Niwa} in particular is given by a jump process, as derived in Section~\ref{derivatsp}. 
Finding a stochastic differential equation whose solution is closely related to the underlying jump process has turned out to be analytically (Section~\ref{SDEjump}) and numerically (Section~\ref{sec:NiwaSDE}) cumbersome. 
However, one can still try to find an SDE which models coagulation and fragmentation dynamics differently and displays the same or a similar equilibrium distribution as the evolution equation~\eqref{forwardequation}. In the following, we consider a \emph{stochastic logistic equation} and its relation to a nearest-neighbour random walk.

\subsubsection{Stochastic logistic equation and gamma distribution} \label{ss:logistic}
From data plotted in \cite[Figure 2]{DLP} and \cite{MJS}, 
one can see that both
the equilibrium profile~\eqref{expansion} for the coagulation and fragmentation model~\eqref{eq:CF1_Niwa} and 
the equilibrium profile~\eqref{Niwaequil} for Niwa's SDE~\eqref{NiwaSDE1} model with $ \bar x = 1 $ and  $\gamma=\frac12$ 
are close to the simple logarithmic size distribution profile
$$ {\Phi}(x)= x^{-1} \exp(-x), $$ 
in the range containing most of the empirical data plotted in
\cite[Figure 5]{N}.
Hence, we can also try to find an SDE, derived from a coagulation-fragmentation model for particles, such that the population distribution
\begin{equation}\label{expdistr}
\rho(x) =  \exp(-x) 
\end{equation}
is the stationary solution of the corresponding Fokker-Planck equation. 
Pursuing this objective, we consider the stochastic logistic equation 
\begin{equation} \label{Stochasticlogistic}
\rmd X_t = rX_t\left(1-\frac {X_t}k\right)\rmd t + \sqrt2 \sigma\,X_t\,\rmd W_t\,, \quad X_0 > 0\,,
\end{equation}
as suggested by Robert May in \cite{Maybook}, and studied in \cite{Polansky}.
If an invariant distribution exists, its density $\rho$ is
the solution of the stationary Fokker-Planck equation
\begin{equation} \label{eq:FPlog}
0 = \frac{\rmd^2}{\rmd x^2}(\sigma^2x^2 \rho) - \frac{\rmd}{\rmd x}\left(
 r x\left(1-\frac xk\right)\rho\right).
\end{equation}
The density must take exactly the form of a power law with exponential cutoff --- a \emph{gamma distribution},
\begin{equation} \label{FPlogsol}
\rho(x) = f(x;\alpha,\beta)= \frac{\beta^\alpha x^{\alpha-1}e^{-\beta x}}{\Gamma(\alpha)},\quad
\alpha = \frac{r}{\sigma^2}-1,\quad \beta = \frac{r\sigma^2}k.
\end{equation}
When $r>\sigma^2$, $\rho$ is integrable on $(0, \infty)$ and $\frac{\beta^a}{\Gamma(\alpha)}$ is exactly the normalization constant.   
Note that we recover the exponential distribution~\eqref{expdistr} 
by choosing $\sigma =1$, $r=2$ and $k=2$. 
When $r\le\sigma^2$, no invariant distribution exists~---~instead one expects
the process to spread out indefinitely as in the case when $r=0$.

In equation~\eqref{Stochasticlogistic}, the degenerate diffusivity proportional to $X_t$ 
prevents the stochastic process from hitting $0$~---~this is a well-known
phenomenon orginating with work of Feller~\cite{Feller1952}. 
In particular, the criterion of Theorem~3 of \cite{Feller1954} states
that the solution $X_t$ of \eqref{Stochasticlogistic}
can hit $0$ if and only if for all $\lambda>0$, all solutions $z(x)$ 
of the ODE 
\begin{equation}
Az = \lambda z\,, \qquad 
A = \sigma^2 x^2 \frac{\rmd^2}{\rmd x^2} +rx \left(1 - \frac{x}{k}\right)\frac{\rmd}{\rmd x}\,,
\end{equation}
on $(0,\infty)$ are bounded in a neighborhood of $0$. 
To determine whether this is the case, one can change variables via 
$y=\log x$ and note that the 
theory of asymptotic behavior of ODEs \cite[Section 3.8]{CL1955}
 allows us to neglect the term $e^y(\rmd z/\rmd y)$ in the limit $y\to-\infty$.
Since the equation 
\[
\sigma^2 \frac{\rmd^2z}{\rmd y^2} + (r-\sigma)\frac{\rmd z}{\rmd y} - \lambda z  = 0
\]
has unbounded solutions, we conclude that
the solution process $X_t$ for \eqref{Stochasticlogistic}
naturally stays in $(0,\infty)$ and 
no further assumption needs to be made about what happens at 0.

We note that the degenerate nature of the diffusion 
does neither reflect Niwa's simulation results, which he used
to estimate diffusivity for the SDE model~\eqref{NiwaSDE1} of 
the merging-splitting dynamics \cite[Figure 2]{N}, nor the diffusion coefficient $c$~\eqref{diffusioncoeff} for the second order approximation of the jump process.
Recall that neither of these approaches delivered results that justifiably model the merging-splitting dynamics described by the jump process. 
In contrast, the logistic SDE model~\eqref{Stochasticlogistic} consistently describes 
group-size fluctuations that occur {\em due to a different mechanism},
namely a geometric Brownian motion with logistic drift.
Consequently, \eqref{Stochasticlogistic} is capable of providing
a rationale for the appearance of the gamma distribution if the group-size dynamics
are governed by a suitable mechanism.

In order to understand better what kind of mechanism could lead to \eqref{Stochasticlogistic},
we note that some basic features observed in
merging-splitting dynamics resemble principles expressed in~\eqref{Stochasticlogistic}: 
The linear multiplicative noise term can be interpreted to correspond with fluctuations increasing with the cluster size due to an increase in ``coagulation and fragmentation interactions". The logistic drift term expresses the dominance of fragmentation for larger sizes and dominance of coagulation for smaller sizes. In the following, we make these notions more precise by describing how 
a classical nearest-neighbor random walk on the lattice (corresponding to small jumps in group size among a discrete set)  formally
corresponds in the continuum limit to the stochastic logistic SDE model~\eqref{Stochasticlogistic}.

\subsubsection{A lattice random walk approximating the stochastic logistic model}
Consider a stochastic process determined exclusively by jumps to the
{\em nearest neighbours} on the lattice $h\mathbb{N}$ with 
small grid size $h$ on $\mathbb{R}_+$.
We let $u_j(t)$ denote the probability of an individual to be in a group of size $jh$ at time $t$, and suppose that the group size can change only by jumps from $jh$ to $jh\pm h$. 
We let $\alpha_j$ denote the
rate of jumps from $jh$ to $jh+h$ and let $\beta_j$ denote the rate of jumps
from $jh$ to $jh-h$. 

The master equation for the corresponding process on the lattice is 
\begin{equation} \label{eq:masterlattice}
\D_t u_j = \alpha_{j-1}u_{j-1}+\beta_{j+1}u_{j+i} - (\alpha_j+\beta_j)u_j,
\quad j=1,2,\ldots.
\end{equation}
We exclude the origin by setting $\alpha_0=0=\beta_1$ and will ignore the boundary henceforth.
We can rewrite equation~\eqref{eq:masterlattice} as
\begin{equation} \label{eq:fluxform1}
\D_t u_j = -F_{j+1/2}+ F_{j-1/2}\,,
\end{equation}
where 
$$F_{j+1/2} := \alpha_j u_j - \beta_{j+1}u_{j+1}$$ 
is the flux from $j$ to $j+1$, and, hence, 
$$F_{j-1/2} = \alpha_{j-1} u_{j-1} - \beta_{j}u_{j}$$ 
is the flux from $j-1$ to $j$.

Let us consider how to choose $\alpha_j$ and $\beta_j$ to approximate
a given SDE
\begin{equation} \label{eq:SDEgen}
\rmd X_t = a(X_t)\,\rmd t + \sqrt2 b(X_t)\,\rmd W_t\,, \quad X_0 > 0\,,
\end{equation}
on $(0,\infty)$. Recall that in the case of~\eqref{Stochasticlogistic} we have
\begin{equation}\label{d:ab_ex}
a(x) = rx - \frac rk x^2,\qquad b(x) = \sigma x.
\end{equation}
The Fokker-Planck equation associated with the SDE~\eqref{eq:SDEgen} is
\begin{equation} \label{eq:FPgen}
\D_t u = -\D_x(a(x)u)+\D_x^2(b(x)u) = \D_x( (2bb'-a)u + b^2 \D_xu).
\end{equation}
Deploying a key idea from numerical analysis,
we write the modified drift as the difference of positive quantities
\[
a-2bb'  = f^+ - f^-\,.
\]
For our example \eqref{d:ab_ex} we can take
\[
f^+(x) = rx, \qquad f^-(x) = \frac rk x^2 + 2\sigma^2 x\,.
\]
Now we discretize the Fokker-Planck equation~\eqref{eq:FPgen}, using upwinding for the drift:
\begin{align*}
\D_t u_j = 
&\ -\frac1h(f^+_j u_j - f^+_{j-1}u_{j-1})
+\frac1h(f^-_{j+1} u_{j+1} - f^-_{j}u_{j}) \\
&\ +\ \frac1{h^2}( b_{j+1/2}^2(u_{j+1}-u_j) - b_{j-1/2}^2(u_j-u_{j-1}) )\,.
\end{align*}
This equation takes the conservative form in \eqref{eq:fluxform1} with 
jumping rates 
\[
\alpha_j = \frac1h f^+_j + \frac1{h^2} b_{j+1/2}^2,
\qquad
\beta_j = \frac1h f^-_j + \frac1{h^2} b_{j-1/2}^2.
\]
In the case of the stochastic logistic equation~\eqref{Stochasticlogistic}, with drift and diffusion coefficients given by~\eqref{d:ab_ex}, we take $x_j = jh$ and $b_{j+1/2}=b(x_j)$, $f_j^+= f^+(x_j)$, $f_j^-= f^-(x_j)$. Hence, we have
\begin{equation} \label{discreterates}
\alpha_j = rj + \sigma^2 j^2, 
\qquad
\beta_j = \frac {rh}k j^2 + 2 \sigma^2 j + \sigma^2(j-1)^2 = 
 \frac {rh}k j^2 +  \sigma^2(j^2+1)\,.
\end{equation}
The jump rates $\alpha_j$ and $\beta_j$ both consist of a term with factor $r$, corresponding with the drift in~\eqref{Stochasticlogistic}, and a term with factor $\sigma^2$, corresponding with the diffusion in~\eqref{Stochasticlogistic}. The terms with factor $\sigma^2$ are all quadratic in the discrete size $j$, as one would expect. The $r$-term in $a_j$ is linear in $j$ and has a stronger relative impact on jumps to the right the smaller $j$ is, as in~\eqref{Stochasticlogistic}. In the rate $\beta_j$ the $r$-term is quadratic in $j$, as in~\eqref{Stochasticlogistic}, implying that both terms contribute to the increasing rate of jumps to the left in the same way. This gives a particular lattice random walk which approximates model~\eqref{Stochasticlogistic} for small $h$.

Note that a probability density $u^{\textnormal{eq}}$ satisfying for all $j \in \mathbb{N}$ the ratio
\begin{equation} \label{equilratio}
\frac{\alpha_j}{\beta_{j+1}} =\frac{u_{j+1}^{\textnormal{eq}}}{u_j^{\textnormal{eq}}}
\end{equation}
is an equilibrium for~\eqref{eq:fluxform1}.
We check this condition for $u_j^{*} := \rho(jh)$ where $\rho$ is the gamma distribution from~\eqref{FPlogsol}. 
First, we observe with a first order Taylor expansion at $h=0$ that
$$ \frac{u_{j+1}^*}{u_j^*} = e^{-\beta h}\left(\frac{j+1}{j}\right)^{\alpha-1} = \left(\frac{j+1}{j}\right)^{\alpha-1} \left( 1- \beta h\right) + \mathcal{O}\left(h^2\right) \,.$$
On the other hand, we expand $\frac{\alpha_j}{\beta_{j+1}}$ at $h=0$ to obtain
$$ \frac{\alpha_j}{\beta_{j+1}} = \frac{j^2 + (\alpha +1)j}{\left(\frac{\beta}{\sigma^4}h +1\right)(j+1)^2 +1} =  \frac{j^2 + (\alpha +1)j}{(j+1)^2 +1} - \frac{\beta h}{\sigma^4} \frac{j^2 + (\alpha +1)j}{(j+1)^2 +1}\frac{(j+1)^2}{(j+1)^2 +1} + \mathcal{O}\left(h^2\right)\,.$$
Recall that choosing $\sigma=1, r=2$ and $k=2$ in~\eqref{Stochasticlogistic} gives the exponential population distribution~\eqref{expdistr}.  In this case we obtain $\alpha =1$ and
\begin{equation} \label{rateratio}
 \frac{u_{j+1}^*}{u_j^*} =  1- \beta h + \mathcal{O}\left(h^2\right), \quad \frac{\alpha_j}{\beta_{j+1}} =  \frac{j^2 + 2j}{j^2+2j+2} - \beta h \left(\frac{j^2 + 2j}{j^2+2j+2} \, \frac{(j+1)^2}{(j+1)^2 +1}\right) + \mathcal{O}\left(h^2\right) \,.
\end{equation}
Observe that for any given $x = jh \in (0, \infty)$ the discrete size $j = \frac{x}{h}$ grows proportionally as $h$ is taken smaller. Therefore, the equilibrium ratio relation~\eqref{equilratio} is satisfied at $x$ in the continuum limit for \eqref{rateratio}, i.e.~when $h \to 0$ and $j = \frac{x}{h} \to \infty$.

This formal derivation indicates that, in terms of the equilibrium density $\rho$, the SDE model~\eqref{Stochasticlogistic} with suitable parameters approximates the nearest neighbour model with jump rates~\eqref{discreterates}. Hence, modelling the coagulation-fragmentation dynamics by the stochastic logistic equation~\eqref{Stochasticlogistic} appears coherent with an underlying locally restricted jump process. This scenario avoids the main problem of the SDE modelling discussed in Sections~\ref{sec:NiwaSDE} and~\ref{SDEjump} where the global aspect of the jump dynamics associated with~\eqref{strongform} cannot be captured by the continuous solution of a stochastic differential equation.

\section{Conclusion}
For coagulation-fragmentation models of the form~\eqref{weakgeneral}, we have derived the evolution equation~\eqref{eq:jumpproc_stronggen} for the population distribution and a formalization of the underlying jump process. The associated algorithm has been validated by showing its accordance with the equilibrium for~\eqref{eq:CF1_Niwa} and its versatility has been demonstrated by also working with different coagulation and fragmentation rates and a numerical study of the respective statistical properties.

Compared to the numerical methods for simulating Niwa-like coagulation and fragmentation models developed and summarised in~\cite{DE}, the jump process algorithm has been shown to be the most versatile and dynamically insightful scheme, by tracking the behaviour of individual trajectories. 
We have seen that, in particular, the rates can be chosen to be random or polynomial. This opens new, potentially more realistic, modelling possibilities that can be further investigated in future work.

Although Niwa's SDE is neither rigorously justified nor particularly well-suited for numerical investigations, it has proven to be an insightful approach to the problem at hand. In Section~\ref{SDEjump} we have mathematically derived an alternative drift-diffusion approximation to 
the jump process whose equilibrium distribution shows similar behavior as the equilibrium for the jump process but does not coincide. To overcome the inherent discrepancy between continuous solutions of SDEs and processes with large jumps, we have indicated an additional possibility using an SDE with degenerate noise (stochastic logistic model) whose equilibria exactly take the form of gamma distributions and which can be related to a nearest-neighbour jump model. A more thorough investigation of that matter is left for future work.

Another future line of investigation could lead to spatialized models where coagulation and fragmentation rates depend on the location of the groups in space. One could imagine several types of spatial inhomogeneties, for example caused by attracting regions with high coagulation activity or volatile regions with high fragmentation probabilities. Such models would have a more direct correspondence with population dynamics and would give rise to new challenges that could be tackled by a jump process approach as discussed in this paper.

\section*{Acknowledgments} 
PD acknowledges support by the Engineering and Physical Sciences  Research Council (EPSRC) under grant no. EP/M006883/1 and EP/P013651/1, by the Royal Society and the Wolfson Foundation through a Royal Society Wolfson Research Merit Award no. WM130048 and by the National Science  Foundation (NSF) under grant no. RNMS11-07444 (KI-Net). PD is on leave  from CNRS, Institut de Math\'ematiques de Toulouse, France. M.E. gratefully acknowledges support from the Department of Mathematics, Imperial College London through a Roth Scholarship and by the German Research Foundation (DFG) via grant SFB/TR109 Discretization in Geometry and Dynamics. This material is based upon work supported by the National
Science Foundation under grants DMS 1812573 (JGL) and DMS 1812609 (RLP) and by the Center for Nonlinear Analysis (CNA) under National Science Foundation PIRE Grant no.\ OISE-0967140,
and by the NSF Research Network Grant no.\ RNMS11-07444 (KI-Net). 

\section*{Data statement} 
No new data were collected in the course of this research.

\end{document}